\documentclass[journal]{IEEEtran}
\usepackage{amsmath,amsfonts}
\usepackage[thmmarks,amsmath]{ntheorem}
\usepackage{algorithmic}
\usepackage{algorithm}
\usepackage{array}
\usepackage[caption=false,font=normalsize,labelfont=sf,textfont=sf]{subfig}
\usepackage{textcomp}
\usepackage{stfloats}
\usepackage{url}
\usepackage{verbatim}
\usepackage{graphicx}
\usepackage{cite}
\usepackage{amssymb}
\usepackage{physics}
\usepackage{epstopdf}
\usepackage{cases}
\usepackage{indentfirst}
\usepackage{amsmath}
\usepackage{booktabs}
\usepackage{color}
\usepackage{balance}
\allowdisplaybreaks[4]

\newtheorem*{proof}{Proof:}
\theoremseparator{:}

\newtheorem{theorem}{Theorem}

\begin{document}
	
	\title{Joint Task Orchestration and Resource Optimization for $\mathbf{SC}^3$ Closed Loop in 6G Networks}
	
	\author{Xinran~Fang,
		Wei~Feng,~\IEEEmembership{Senior Member,~IEEE,} Yanmin Wang, Yunfei~Chen,~\IEEEmembership{Fellow,~IEEE,}\\ 
		Baoquan Ren, Ning~Ge, and Shi~Jin,~{\IEEEmembership{Fellow,~IEEE}}
		\thanks{Xinran~Fang, Wei~Feng, and Ning Ge are with the Department of Electronic Engineering, State Key Laboratory of Space Network and Communications, Tsinghua University, Beijing 100084, China (e-mail: {fxr20}@mails.tsinghua.edu.cn, {fengwei}@tsinghua.edu.cn, and  {gening}@tsinghua.edu.cn).
			
			Yanmin Wang is with the School of Information Engineering, Minzu University of China, Beijing 100081, China~(email: wangyanmin@muc.edu.cn).
			
			Yunfei Chen is with the Department of Engineering, University of Durham,
			DH1 3LE Durham, UK. (e-mail: yunfei.chen@durham.ac.uk).
			
			Baoquan Ren is with the China Electronic System Engineering Company, Beijing 100141, China~(email: renbq88@126.com).
			
			Shi Jin is with the National Mobile Communications Research Laboratory, Southeast University, Nanjing 210096, China (e-mail: jinshi@seu.edu.cn).
		}
	}

	\maketitle
	
\begin{abstract}
In hazardous environments, sensors and actuators can be deployed to \textit{see} and \textit{operate} on behalf of humans, enabling safe and efficient task execution. Functioning as a \textit{neural center}, the edge information hub (EIH), which integrates communication and computing capabilities, coordinates these sensors and actuators into sensing-communication-computing-control ($\mathbf{SC}^3$) closed loops to enable autonomous operations. From a system-level optimization perspective, this paper addresses the problem of joint sensor-actuator pairing and resource allocation across multiple $\mathbf{SC}^3$ closed loops. To tackle the resulting mixed-integer nonlinear programming problem, we develop a learning-optimization-integrated actor-critic (LOAC) framework. In this framework, a deep neural network-based actor generates pairing candidates, while an optimization-based critic subsequently allocates communication and computing resources. The actor is then iteratively refined through feedback from the critic. Simulation results demonstrate that the LOAC framework achieves near-optimal solutions with low computational complexity, offering significant performance gains in reducing control cost.
\end{abstract}

\begin{IEEEkeywords}
Autonomous operation, learning-optimization-integrated actor-critic framework, sensing-communication-computing-control ($\mathbf{SC}^3$) closed loop.
\end{IEEEkeywords}

\section{Introduction}
\subsection{Background and Motivation}
The demand for autonomous operations is growing rapidly. In disaster rescue, for instance, numerous rescue robots have been developed and deployed in post-disaster scenarios such as earthquakes, firefighting, and nuclear incidents. In April 2024, Boston Dynamics unveiled the second generation of their humanoid robot, {\textit{Atlas}}, capable of performing various fundamental yet potentially life-saving tasks, including lifting barriers, opening doors, and activating switches \cite{atlas}. However, a single robot is typically insufficient to tackle intricate tasks in dynamic environments. When multiple robots are dispatched, they require a nervous system-like communication network to coordinate with each other. To fully unlock the potential of robotic systems, sixth-generation (6G) networks aim to support collaborative robotic operations \cite{Lei}.

Taking disaster rescue as an example, as shown in Fig. \ref{fig1}, a 6G-enabled autonomous operation system typically operates through the following two phases:
\begin{itemize}
	\item \textit{Emergency Detection and Preparation:} \
	Upon detecting an emergency event such as a wildfire, preset sensors report to the nearest disaster management center. The center then activates all available sensors, including both pre-installed devices and those deployed in real time (e.g., via aerial delivery). These sensors collect data from affected areas and transmit it back to the center. Based on its assessment of disaster severity, the center devises a dispatch plan specifying the number of rescue robots to deploy and the configuration of the airborne edge information hub (EIH)~\cite{Lei2}, including its communication and computing modules.
	\item \textit{On-site Deployment and Rescue Operation:} \
	Once the airborne EIH and field robots arrive at the designated site, the EIH positions itself and establishes backhaul connections with the cloud center via satellite. It then collects sensor data, conducts analyses, and disseminates control commands to actuators. This setup creates sensing-communication-computing-control ($\mathbf{SC}^3$) closed loops, each autonomously executing a distinct task. The EIH serves simultaneously as the access point connecting distributed sensors and actuators, and as the computing unit that transforms raw sensing data into executable control commands.
\end{itemize}

An $\mathbf{SC}^3$
closed loop is analogous to the human reflex arc: it consists of a sensor, a computing center, an actuator, and the corresponding uplink (sensor-to-EIH) and downlink (EIH-to-actuator) (UL \& DL) communications. In each $\mathbf{SC}^3$
cycle, the sensor perceives the environment and target, the UL transmits the sensed data to the computing center, where the data is processed to generate control commands, which are then forwarded to the actuator via the DL for execution.
To establish an $\mathbf{SC}^3$
closed loop, the EIH must simultaneously connect one sensor and one actuator. In this context, sensor-actuator pairing emerges as a novel challenge absent from traditional communication networks, which typically establish links with individual users rather than jointly connecting two nodes to form a closed loop. Furthermore, sensor-actuator pairing and resource allocation are inherently coupled: sensor selection affects UL communication and computing requirements, while actuator selection influences DL communication. Conversely, available communication and computing resources determine which pairings are feasible or efficient. This interplay introduces both opportunities and challenges for designing 6G networks to support autonomous operations.

Therefore, this paper explores task-resource bidirectional adaptation regime and proposes a joint sensor-actuator pairing and resource allocation scheme. The joint optimization of discrete pairing decisions and continuous resource variables leads to a complex mixed-integer nonlinear programming (MINLP) problem. Efficiently solving this problem is critical for enabling real-time system reconfiguration in dynamic environments. In the following, we review related works to identify research gaps and motivate our proposed approach.

\begin{figure*} [t]
\centering
\includegraphics[width=0.8\linewidth]{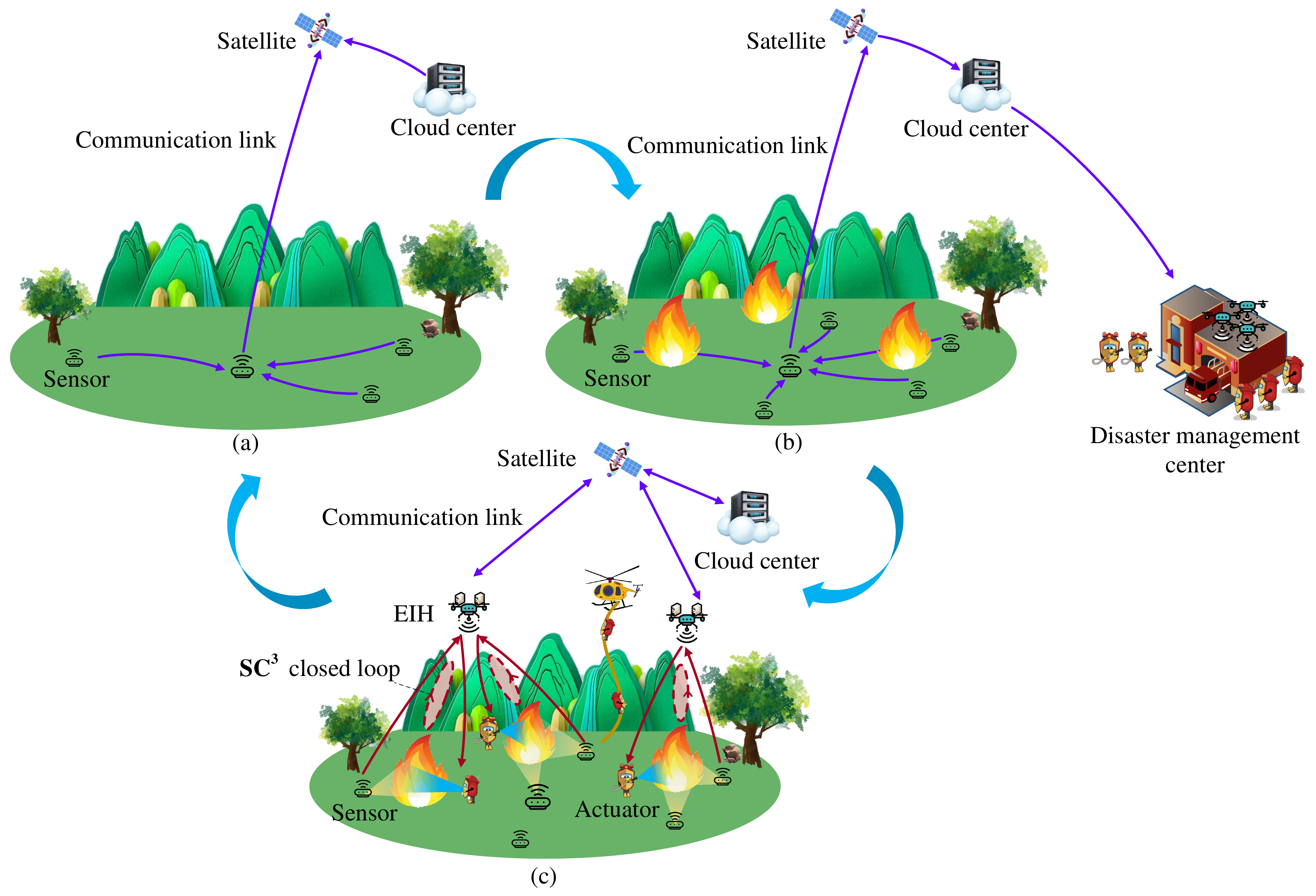}
\caption{Illustration of a 6G-enabled autonomous operation system for disaster rescue: (a) normal state, (b) emergency detection and preparation, and (c) on-site deployment and rescue operation.}
\label{fig1}
\end{figure*}

\subsection{Related Works}
\subsubsection{Closed-Loop Optimization}
In wireless control systems (WCSs) \cite{Park}, the interplay between wireless communications and control performance has been widely studied. Given that the inherent uncertainty of wireless communications significantly affects closed-loop control performance, fundamental communication requirements for control stabilization have been established. For linear time-invariant (LTI) systems, Tatikonda \emph{et al.} demonstrated that stabilizing a system requires a minimum communication rate equal to the sum of the logarithms of the unstable eigenvalues of the state matrix \cite{Tatikonda}. Building upon this result, Sahai \emph{et al.} introduced the concept of anytime capacity, which characterizes system stability under noisy observations, delayed actions, and the absence of feedback \cite{Sahai}. Kostina \emph{et al.} further derived a lower bound on the linear quadratic regulator (LQR) cost under rate constraints, quantifying the minimum information required to achieve a target LQR cost \cite{Kostina}. Li \emph{et al.} conducted closed-loop timing analysis by considering the randomness of computation, compression, and communication latency \cite{Li1}. Collectively, these results establish the theoretical foundation for WCS design.

To effectively utilize limited communication resources, control-aware communication schemes have been proposed that allocate wireless resources based on control objectives. To reduce LQR cost, Lei \emph{et al.} developed a control-oriented DL power allocation scheme and demonstrated its advantages over traditional water-filling methods \cite{Lei3}. Fang \emph{et al.} extended this design to time-sensitive scenarios by jointly optimizing transmit power and symbol blocklength \cite{Fang1}. Chang \emph{et al.} proposed a control-aware scheme that maximizes spectral efficiency subject to a required control convergence rate \cite{Chang}. Additionally, Ali \emph{et al.} employed deep reinforcement learning (DRL) to co-design communication and control policies, optimizing the control sampling period, blocklength, and packet error rate \cite{Ali}. 

More recently, research efforts have expanded from control closed loops to $\mathbf{SC}^3$
closed loops. For single-loop optimization, our prior work introduced a loop-level metric termed the closed-loop negentropy rate (CNER) and identified the task-level trade-off between UL and DL \cite{Fang2}. This study was subsequently extended to multi-loop settings, where a joint intra-loop and inter-loop resource allocation scheme was devised to minimize the sum LQR control cost \cite{Fang3}. Furthermore, Meng \emph{et al.} applied effective capacity theory to model the $\mathbf{SC}^3$
closed loop, proposing a joint control and resource allocation scheme that efficiently utilizes wireless bandwidth while ensuring control stability \cite{Meng}. Ying \emph{et al.} studied an automated guided vehicle system, jointly adjusting control cycles and communication schedules to meet control targets \cite{Ying}. Unlike these studies, which typically assume known system parameters, Han \emph{et al.} investigated a novel problem incorporating system identification (SI) into the design, proposing a scheme that minimizes communication energy while satisfying SI requirements \cite{Han}. You \emph{et al.} developed a sensing data collection strategy that minimizes latency under control performance constraints \cite{You}, proposing a water-sampling algorithm for UL bandwidth allocation as an extension to the classical water-filling method.

These studies provide valuable insights into closed-loop optimization. However, they all assume a predetermined loop configuration and focus solely on resource allocation. This limits system adaptability in dynamic environments, particularly when link conditions cannot support communication between connected nodes. The problem of jointly optimizing sensor-actuator pairing and resource allocation remains largely unexplored.

\subsubsection{Task-Oriented Communications}
Task-oriented communication has recently attracted significant attention in 6G network development, particularly for supporting vertical applications. Many studies have focused on learning-task-oriented communications, where communication strategies are optimized for learning metrics such as inference accuracy and convergence rate \cite{Shi}. For example, Shao \emph{et al.} applied information bottleneck theory to devise a task-oriented communication framework in which the transmitter extracts image features and transmits them to the receiver for inference \cite{Shao}. This approach maintains high inference accuracy while reducing communication overhead by transmitting only essential feature information. The authors subsequently extended this framework to video analysis tasks, leveraging temporal correlations between adjacent frames to alleviate communication burdens \cite{Shao2}. Building on similar principles, Diao \emph{et al.} investigated an autonomous driving scenario and proposed a task-oriented joint source-channel coding strategy coupled with trajectory optimization to reduce energy consumption \cite{Diao}.

Moving beyond intermediate learning metrics, recent works have sought to optimize ultimate task performance directly. For instance, Tung \emph{et al.} employed DRL to jointly optimize robot actions alongside communication strategies (i.e., channel coding and modulation) for guided robotic tasks, demonstrating superior performance over traditional schemes that separately design communication and control policies \cite{Tung}. Similarly, Mostaani \emph{et al.} considered a multi-agent system and used DRL to simultaneously learn the value of information, each agent's quantization policy, and the controller's policy \cite{Mostaani}. Focusing on communication among embodied agents, Zhang \emph{et al.} provided a comprehensive review of recent task-oriented communication approaches, summarizing major challenges and solution strategies for three-stage task-oriented design \cite{Zhang1}.

These studies represent early attempts to embed task requirements into communication design. However, they focus primarily on individual communication links rather than $\mathbf{SC}^3$
closed loops. Task-oriented closed-loop optimization remains largely unexplored.

\subsubsection{Mixed-Integer Nonlinear Programming}
In various wireless communication scenarios, optimization problems involving user association or pairing are typically formulated as MINLP problems. The inherent complexity of these problems stems from the deep coupling among UL and DL communications, edge computing, and physical control. To address such problems, diverse modeling and solution frameworks have been proposed. For instance, to ensure hyper-reliable low-latency communication in wireless control systems, Li \emph{et al.} optimized continuous resources via alternating optimization, while addressing integer scheduling variables through exhaustive search or flow-shop algorithms \cite{Li2026Joint}. In mobile edge computing, Feng \emph{et al.} decoupled the MINLP problem into continuous and integer subproblems, solving them via dual decomposition and matching-theoretic algorithms, respectively \cite{Feng2021Joint}. Sun \emph{et al.} jointly optimized UL and DL communication and computing, employing exhaustive search to obtain optimal sensor-actuator pairing \cite{Sun2022Joint}. In \cite{Tran}, the authors adopted a decoupled methodology, utilizing a heuristic algorithm to determine discrete offloading strategies, followed by convex and quasi-convex optimization techniques for continuous resource allocation. Gu \emph{et al.} proposed a communication-computation-aware association mechanism by relaxing binary variables into continuous probabilities within a meta-analytical framework \cite{Gu2023Communication}.

Ding \emph{et al.} developed a DRL approach to jointly optimize offloading strategies, transmit power, and CPU frequencies to maximize computing efficiency \cite{Ding}. Focusing on wireless-powered mobile edge computing networks, Huang \emph{et al.} \cite{Huang} proposed a novel DRL-based online offloading framework. The authors employed a deep neural network to learn offloading decisions, coupled with an optimization-based module for continuous resource allocation. Li \emph{et al.} subsequently applied the same learning-optimization integrated approach to solve combinatorial computation offloading problems \cite{Li}. For full-duplex networks, Sekander \emph{et al.} addressed decoupled UL and DL association by transforming the problem into a geometric programming framework supplemented by a distributed matching game \cite{Sekander2017Decoupled}. Zhang \emph{et al.} maximized max-min fairness using continuous relaxation and greedy rounding \cite{Zhang2019MaxMin}.

In summary, exact algorithms such as exhaustive search \cite{Sun2022Joint, Li2026Joint} provide global optimality but suffer from high computational complexity, rendering them unsuitable for real-time autonomous operations. For improved efficiency, decoupled optimization schemes are widely adopted, typically combining heuristics or matching games for discrete variables with mathematical optimization for continuous variables \cite{Tran}. However, these approaches fail to fully exploit the intrinsic coupling between integer decisions and continuous resource allocation, often leading to performance degradation. To address this limitation, joint optimization frameworks have emerged, broadly classified into model-based and model-free approaches. Model-based methods typically employ continuous relaxation techniques, where integer variables are relaxed into continuous space. Since the resulting relaxed problems are generally non-convex, techniques such as successive convex approximation (SCA) and geometric programming are employed \cite{Sekander2017Decoupled, Zhang2019MaxMin, survey3}. However, these methods often require substantial iteration overhead and incur non-trivial performance loss during binary recovery. Alternatively, model-free methods, particularly those based on DRL, map network states directly to decisions. Such approaches can achieve near-optimal solutions once a well-designed neural network is trained \cite{Ding}. Nevertheless, applying DRL directly to MINLP problems requires discretization of continuous variables, leading to factorial explosion in the action space. Consequently, learning-optimization integrated frameworks have emerged as a highly effective alternative \cite{Li, Huang}. By employing neural networks for discrete decisions and leveraging mathematical optimization for continuous resource allocation, these frameworks successfully bridge the gap between computational speed and solution optimality for real-time applications.

\subsection{Main Contributions and Organization}
In this paper, we consider a 6G-enabled autonomous operation system deployed in post-disaster areas. To maximize system utility for closed-loop control, we propose a system-level closed-loop optimization scheme and develop a learning-optimization-integrated actor-critic (LOAC) framework. The main contributions are summarized as follows:
\begin{enumerate}
	\item We formulate an autonomous operation system comprising $\mathbf{SC}^3$
	closed loops. To maximize system utility in task execution, we jointly optimize sensor-actuator pairing and communication/computing resource allocation. This joint design enables dynamic task-resource bidirectional adaptation to changing conditions.
	\item We model the task-execution process as a closed-loop control process and introduce a loop-level metric termed CNER. Building upon this, the sum LQR cost is minimized by jointly optimizing pairing decisions, UL and DL bandwidth and power, as well as CPU frequency.
	\item To solve the resulting MINLP problem, we develop the LOAC framework. The deep neural network (DNN)-based actor generates sensor-actuator pairing solutions, while the optimization-based critic allocates continuous resources given the pairing decisions. Feedback from the critic is utilized to iteratively update the DNN, enabling the LOAC framework to provide near-optimal solutions with low computational complexity.
\end{enumerate}

The remainder of this paper is organized as follows. Section~\ref{section 2} introduces the autonomous operation system and the $\mathbf{SC}^3$ closed-loop model. Section~\ref{section 3} presents the joint sensor-actuator pairing and resource allocation scheme along with the LOAC framework. Section \ref{section 4} presents simulation results and discussion. Section~\ref{section 5} concludes the paper.

\textbf{Notation}: We use bold lowercase letters, bold uppercase letters, and calligraphic uppercase letters to denote vectors, matrices, and sets, respectively. $\mathbf{0}_n$ denotes the $n \times n$ zero matrix, $\mathbf{I}_n$ denotes the $n \times n$ identity matrix, and $\mathbb{R}^{n \times n}$ denotes the set of $n\times n$ real matrices. Furthermore, $\det(\mathbf{A})$ represents the determinant of matrix $\mathbf{A}$, $|a|$ denotes the absolute value of scalar $a$, and $|\mathcal{A}|$ denotes the cardinality of set $\mathcal{A}$. The complex Gaussian distribution with zero mean and variance $\sigma^2$ is denoted as $\mathcal{CN}(0,\sigma^2)$.

\section{Autonomous Operation and $\mathbf{SC}^3$ Closed-Loop Model}
\label{section 2}
\begin{figure} [t]
\centering
\includegraphics[width=0.9\linewidth]{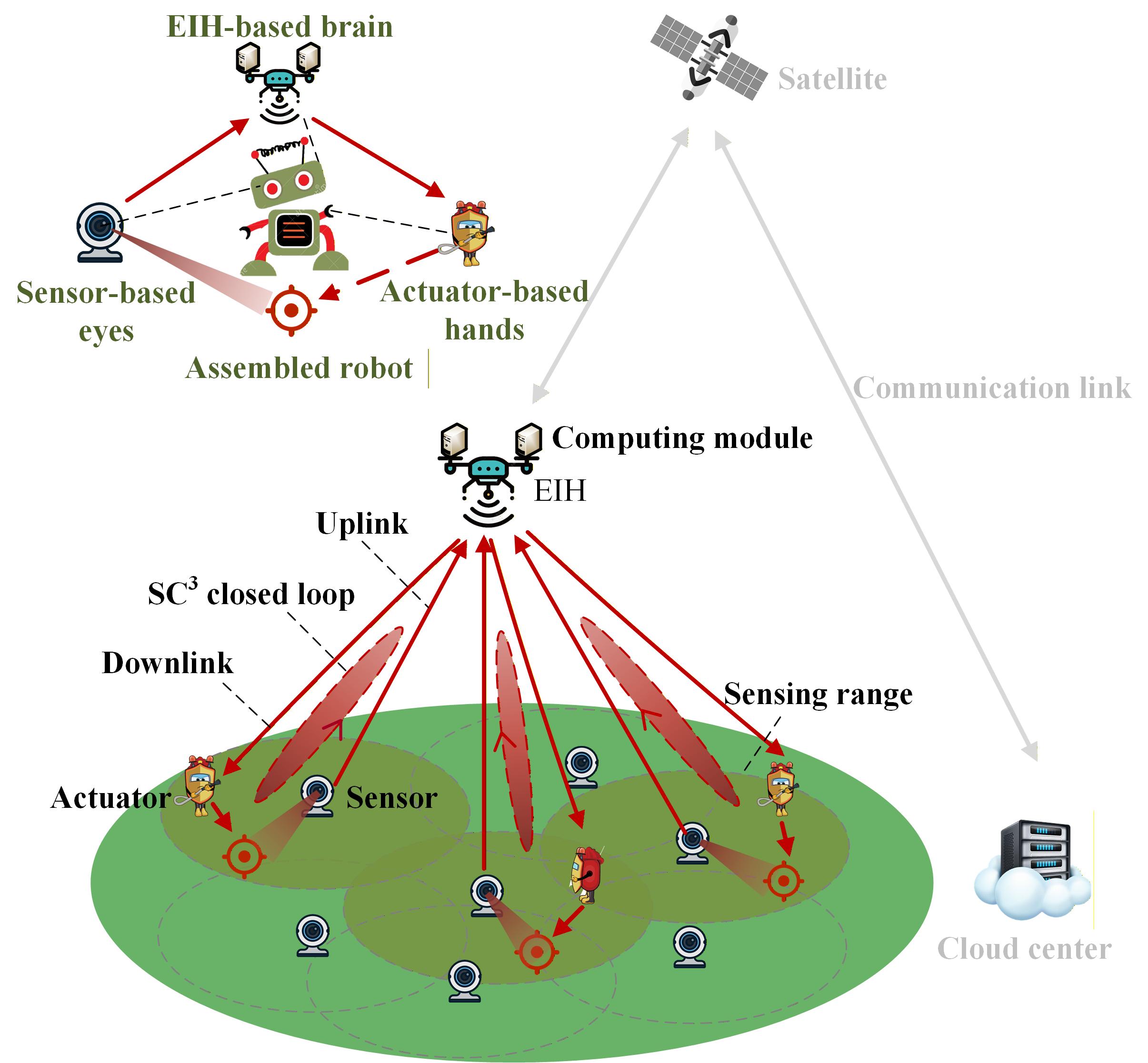}
\caption{{Illustration of a 6G autonomous operation system, where the EIH orchestrates distributed sensors and actuators into integrated $\mathbf{SC}^3$ closed loops. As depicted in the upper left, the $\mathbf{SC}^3$ closed loop is conceptually analogous to an \emph{assembled robot}: operating as a unified entity, it gathers environmental data through its sensor-based \emph{eyes}, processes information within its EIH-based \emph{brain}, and executes precise control commands through its actuator-based \emph{hands}.}}
\label{model}
\end{figure}

As illustrated in Fig. \ref{model}, this study considers a disaster rescue scenario where hazardous environments preclude human intervention. In such contexts, a 6G-enabled autonomous operation system offers a promising solution. Given the potential absence or failure of terrestrial infrastructure, an airborne EIH is rapidly deployed to the site and pairs distributed sensors and actuators into $\mathbf{SC}^3$ closed loops, conceptually analogous to assembled robots. Within each loop, the EIH functions as the computational \emph{brain}, while sensors and actuators serve as the perceptual \emph{eyes} and operational \emph{hands}, respectively. Operating as an integrated unit, each robot gathers environmental data through its sensor-based \emph{eyes}, processes information within its EIH-based \emph{brain}, and executes control commands through its actuator-based \emph{hands}.
For architectural completeness, the system diagram includes the satellite and cloud center, which provide global support such as tracking, telemetry, and high-level control guidance for the EIH. However, the technical contribution of this paper focuses on modeling and optimizing the edge autonomous system. To visually delineate this scope in Fig. \ref{model}, the satellite and cloud center are rendered in gray.

The whole control task is decomposed into $K$ independent subtasks, which are subsequently assigned to $K$ $\mathbf{SC}^3$ closed loops. Taking the $k$-th subtask as an example, the controlled system is modeled as an LTI system,
\begin{equation}
\label{sta_elv}
\mathbf{x}_{k,i+1} = \mathbf{A}_k\mathbf{x}_{k,i}+\mathbf{B}_k\mathbf{u}_{k,i}+\mathbf{v}_{k,i},
\end{equation}
where $\mathbf{x}_{k,i}\in\mathbb{R}^{n_k\times1}$ denotes the system state at time index $i$ ($n_k$ is the dimension of the system state), $\mathbf{u}_{k,i}\in\mathbb{R}^{m_k\times1}$ denotes the control input ($m_k$ is the dimension of the control input), and $\mathbf{v}_{k,i}\in\mathbb{R}^{n_k\times1}$ denotes the system noise with covariance matrix $\mathbf{\Sigma}_{\mathbf{v}_k}$. In addition, $\mathbf{A}_k\in\mathbb{R}^{n_k\times n_k}$ is the state matrix, and $\mathbf{B}_k\in\mathbb{R}^{n_k\times m_k}$ is the input matrix. The intrinsic entropy, i.e., $\log|\det\mathbf{A}_k|$, measures the system instability \cite{Tatikonda}. Systems with a higher intrinsic entropy are inherently more challenging to control. We use LQR cost to evaluate the control performance of these $\mathbf{SC}^3$ closed loops, which measures the cumulative cost of state deviations and control inputs over time:
\begin{equation}
l_k=\limsup\limits_{N\rightarrow \infty}\mathbb{E} \left[ \frac{1}{N}\sum_{i=1}^{N} \left(\mathbf{x}_{k,i}^\text{T}\mathbf{Q}_k\mathbf{x}_{k,i} +\mathbf{u}_{k,i}^\text{T}\mathbf{R}_k\mathbf{u}_{k,i}\right) \right],
\end{equation}
where $l_k$ denotes the LQR cost, $\mathbf{Q}_k\in\mathbb{R}^{n_k\times n_k}$ and $\mathbf{R}_k\in\mathbb{R}^{m_k\times m_k}$ are positive semi-definite weighting matrices. A smaller $l_k$ indicates better control performance with lower input cost.

We assume that there are $S$ sensors and $K$ actuators. Each actuator is positioned within its designated control region, while sensors are distributed more densely across the disaster area. Given that $S>K$, the EIH must select $K$ sensors from the total $S$ and pair them with actuators to establish $K$ $\mathbf{SC}^3$ closed loops. For each sensor $s$, we consider the following attributes that influence selection:
\begin{equation}
\{\mathbf{q}_s, r_s, h^u_s, p^{u}_{s,\max}, D^s_s, \gamma_s, \rho_s\}, \label{1}
\end{equation}
where $\mathbf{q}_s\in\mathbb{R}^{2\times1}$ denotes the sensor's location, $r_s$ denotes the sensing range, $h^u_s$ denotes the UL channel gain, $p^{u}_{s,\max}$ denotes the maximum transmit power, and $D^s_s$ (bits/$\mathbf{SC}^3$ cycle) denotes the sensing data volume per $\mathbf{SC}^3$ cycle, termed the sensing cycle rate. {Furthermore, $\gamma_s$ (CPU cycles/bit) characterizes the computational intensity of sensing data, representing the CPU cycles required per bit of raw sensing input \cite{Sun2022Joint}. Meanwhile, $\rho_s \in [0,1]$ denotes the information extraction ratio \cite{Fang3}, modeling the transformation of high-dimensional raw sensing data into concise actionable control commands. For instance, in visual-based firefighting control, raw image frames are processed at the EIH to extract specific commands such as steering angles or velocity vectors for actuators. In practice, $\gamma_s$ and $\rho_s$ can be instantiated for specific tasks based on sensing data types and computing models. Refining internal computing models to explicitly instantiate these parameters for specialized unmanned tasks remains an important direction for future research.}

For the actuator $k$, we characterize its attributes by
\begin{equation}
\{\mathbf{q}_k, h_k^d, \mathbf{A}_k,\mathbf{B}_k,\mathbf{Q}_k,\mathbf{R}_k,\mathbf{\Sigma}_{\mathbf{v}_k}\}, \label{2}
\end{equation}
which represent the location, DL channel gain, and the parameters of controlled system. 
In the following, we introduce the $\mathbf{SC}^3$ closed loop model in detail, and show how the sensor and actuator traits impact the closed-loop performance.

\begin{figure*} [t]
\centering
\includegraphics[width=1\linewidth]{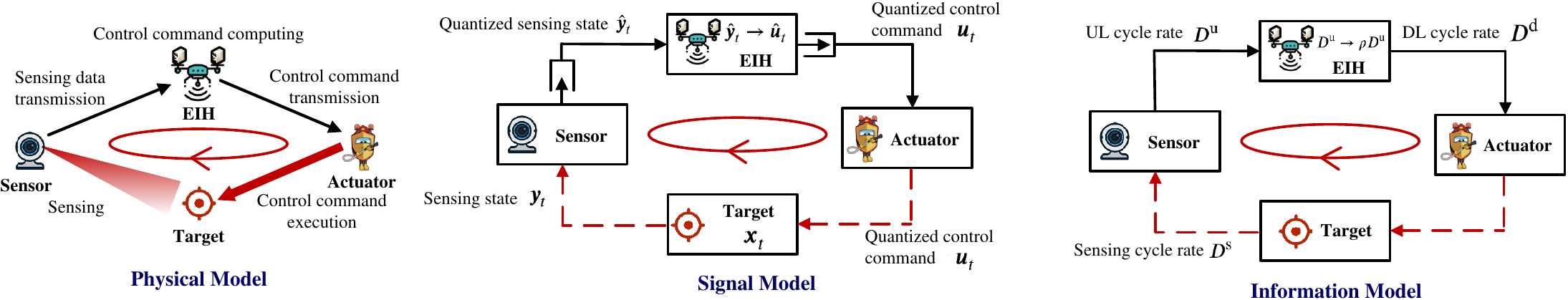}
\caption{Illustration of the $\mathbf{SC}^3$ closed-loop model from three levels: physical level, signal level, and information level.}
\label{diagram}
\end{figure*}

Firstly, the sensor $s$ can be paired with actuator $k$ only when the actuator is located within its sensing range.
Accordingly, we denote the effective sensor set of actuator $k$ as $\mathcal{S}_k$, which contains all the sensor indexes that satisfy the sensing range constraint as
\begin{equation}
\label{4}
\left\{
\begin{aligned}
	&s\in\mathcal{S}_k    &\|\mathbf{q}_s-\mathbf{q}_k\|\leqslant r_s, \\
	&s\notin\mathcal{S}_k  &  \text{otherwise}. \\
\end{aligned}
\right.
\end{equation}
For the communication process, we use composite channel models for both UL\&DL transmissions. All channels are assumed to be independent and identically distributed (i.i.d.). The UL channel gain of sensor $s$ and the DL channel gain of actuator $k$ are calculated by
\begin{equation}
\label{5}
\begin{aligned}
	&	h^u_s=\alpha^u_s(\beta^u_s)^{\frac{1}{2}}, \\
	&	h^d_k=\alpha^d_k(\beta^d_k)^{\frac{1}{2}}.
\end{aligned}
\end{equation}
For simplicity, we omit the subscript and use $\alpha$ and $\beta$ to introduce their expressions. Specifically, $\alpha \sim \mathcal{CN}(0,1)$ denotes the small-scale fading, and $\beta$ denotes the large-scale fading, which is calculated by \cite{Chen}
\begin{equation}
\label{6}
\begin{aligned}
	\beta[\text{dB}]&=\frac{A}{1+ae^{-b(\rho-a)}}+B,\\
	\beta&=10^{-\frac{\beta [\text{dB}]}{10}},
\end{aligned}
\end{equation}
with
\begin{subequations}
\begin{align}
	&A=\eta_{LOS}-\eta_{NLOS}, \\
	&B=20\log_{10}(d)+20\log_{10}\big(\frac{4\pi f_c}{c}\big)+\eta_{NLOS}, \\
	&\rho=\frac{180}{\pi}\arcsin\big(\frac{h}{d}\big),
\end{align}
\end{subequations}
where $d$ (m) denotes the distance between the EIH and the sensor or the actuator,  $h$ (m) denotes the height of the EIH, $c$ (m/s) denotes the speed of light, $f_c$ (Hz) denotes the carrier frequency, and $\eta_{LOS}$, $\eta_{NLOS}$, $a$, and $b$ are environment-dependent constants.
Let $t^u_k$ and $t^d_k$ denote the time allocated for UL and DL transmissions, respectively. {The volumes of UL\&DL transmitted data in one $\mathbf{SC}^3$ cycle, denoted as $D^u_k\ (\text{bits/$\mathbf{SC}^3$ cycle})$ and $D^d_k\ (\text{bits/$\mathbf{SC}^3$ cycle})$, can be calculated as
\begin{equation}
\label{8}
\begin{aligned}
	&D^u_k\leqslant \sum_{s=1}^Sa_{sk}t^u_kB_k\log_2(1+\frac{p^u_s|h^u_s|^2}{N_0B_k}), \\
	&D^d_k\leqslant t^d_kB_k\log_2(1+\frac{p^d_k|h^d_k|^2}{N_0B_k}),
\end{aligned}
\end{equation}
}
where $D^u_k$ and $D^d_k$ are referred to as UL\&DL cycle rates, $N_0$ denotes the noise power spectrum density, $p^u_s$ denotes the transmit power of sensor $s$, $p^d_k$ denotes the allocated DL transmit power to actuator $k$, and $B_{k}$ is the allocated bandwidth, which is shared between the UL\&DL as the two links operate in a time-division manner. The binary variable $a_{sk}$ indicates whether sensor $s$ is paired with actuator $k$
\begin{equation}
\left\{
\begin{aligned}
	&a_{sk}=1,\quad \text{sensor $s$ is paired with actuator $k$},\\
	&a_{sk}=0, \quad \text{sensor $s$ is not paired with actuator $k$}.
\end{aligned}
\right.
\end{equation}

To ensure a one-to-one matching between sensors and actuators within each $\mathbf{SC}^3$ closed loop, $a_{sk}$ satisfies
\begin{equation}
\begin{aligned}
	&\sum\limits_{k=1}^K a_{sk}\leqslant1,\ s=1,\dots,S, \\
	&\sum\limits_{s=1}^Sa_{sk}=1, \ k=1,\dots,K. \\
\end{aligned}
\end{equation}

In the EIH, raw data from sensors undergoes processing prior to control decision-making. The data volume received at the EIH is constrained by both the sensing cycle rate and the UL cycle rate, i.e., $\min\{ \sum\limits_{s=1}^Sa_{sk}D^s_s, D^u_k\}, \forall k$. To focus on communication and computing resource allocation, we assume that the sensing cycle rates of all sensors exceed their respective UL cycle rates. Under this assumption, the data volume received at the EIH equals the UL cycle rate, $D^u_k$, and the effective information output is $\sum\limits_{s=1}^S a_{sk}\rho_s D^u_k$. This process can be formulated as
\begin{equation}
\label{11}
D^u_k\rightarrow \sum_{s=1}^Sa_{sk}\rho_s D^u_k.
\end{equation}
Furthermore, computation must be completed within a strict latency budget. Let $f_k$ denote the CPU frequency allocated to loop $k$. The computing time is constrained by a specified limit $t^c_k$:
\begin{equation}
\label{12}
\frac{\sum\limits_{s=1}^S a_{sk}\gamma_s D^u_k}{f_k}\leqslant t^c_k,
\end{equation}
where $\gamma_s$ (CPU cycles/bit) represents the computational complexity for processing data from sensor $s$.

Then, the information extracted during computing is transmitted to the actuator via the DL. The data received by the actuator is denoted as $D^{\mathbf{SC}^3}_k$, which is given by
\begin{equation}
\label{eq:cner_definition}
D^{\mathbf{SC}^3}_k = \min\left\{\sum_{s=1}^S a_{sk}\rho_s D^u_k, \ D^d_k\right\}.
\end{equation}
We denote $D^{\mathbf{SC}^3}_k$ as the CNER. As its name implies, CNER quantifies the entropy reduction of the controlled system achieved through one $\mathbf{SC}^3$ cycle. As a loop-level metric, CNER unifies the performance of UL, computing, and DL functions, reflecting the efficiency of the entire $\mathbf{SC}^3$ closed loop rather than any single component. {Within the CNER formulation, $\rho_s$ 
quantifies the effective information contained in the transmitted UL data. The resulting asymmetric demands of UL\&DL transmissions underscore the necessity of a joint closed-loop design to optimally balance UL\&DL, and computing capabilities.}

Furthermore, the modeling of the $\mathbf{SC}^3$ closed loop can be classified into three hierarchical levels: physical, signal, and information. As illustrated in Fig. \ref{diagram}, the physical level corresponds to the real world and provides the most accurate representation of physical processes; however, design schemes at this level are often computationally prohibitive. The signal level focuses on specific system states and control commands exchanged within each $\mathbf{SC}^3$ cycle. While suitable for engineering design of quantization, coding, and control algorithms, this level lacks a unified performance measure across sensing, communication, computing, and control functions. The CNER we introduce is modeled at the information level, which appropriately abstracts these heterogeneous functions and integrates them from an information-theoretic perspective. This abstraction enables unified quantitative characterization of closed-loop performance.

According to \cite{Kostina}, the LQR cost is lower bounded by a function of CNER
\begin{equation}
\label{13}
l_k \geqslant \frac{n_k N \!\left( \mathbf{v}_k\right)|\det \mathbf{M}_k|^\frac{1}{n_k}} {2^{\frac{2}{n_k}(D^{\mathbf{SC}^3}_k-\log_2|\det \mathbf{A}_k|)}-1}+\mathrm{Tr}\left( \mathbf{\Sigma}_{\mathbf{v}_k}\mathbf{S}_k\right),
\end{equation}
where $N(\mathbf{v}_k)$ denotes the differential entropy of $\mathbf{v}_k$,
\begin{equation}
N(\mathbf{v}_k)=\frac{1}{2\pi e}e^{-\frac{2}{n_k}\int_{\mathbb{R}^{n_k}}g({\mathbf{v}_k})\log ( g({\mathbf{v}_k}))\mathrm{d}{\mathbf{v}_k}}, 
\end{equation}
with $g({\mathbf{v}_k})$ representing the probability density function of $\mathbf{v}_k$. The matrices $\mathbf{M}_k$ and $\mathbf{S}_k$ are obtained from the following discrete-time algebraic Riccati equations,
\begin{equation}
\label{17}
\begin{aligned}
	\mathbf{S}_k & = \mathbf{Q}_k + \mathbf{A}_k^{\mathrm{T}}\left(\mathbf{S}_k- \mathbf{M}_k\right) \mathbf{A}_k,\\
	\mathbf{M}_k & = \mathbf{S}_k \mathbf{B}_k \left( \mathbf{R}_k + \mathbf{B}_k^{\mathrm{T}} \mathbf{S}_k \mathbf{B}_k\right)^{-1} \mathbf{B}_k^{\mathrm{T}}\mathbf{S}_k.
\end{aligned}
\end{equation}

In this paper, we adopt \eqref{13} as the optimal LQR cost subject to the CNER constraint. It is important to clarify that while the theoretical lower bound in \eqref{13} is derived based on a long-term average information rate, our defined CNER evaluates the instantaneous rate of a single $\mathbf{SC}^3$
cycle. To bridge this timescale gap, we assume a quasi-static communication environment, given that air-to-ground channels experience limited variations. Under this assumption, the CNER remains approximately constant over hundreds of control cycles and can thus be used to estimate LQR cost performance. Furthermore, in the presence of large-scale environmental variations, our proposed scheme dynamically linearizes the system at the new operating point, triggering a fresh round of joint resource reallocation to adapt to updated channel conditions and system dynamics.

On this basis, we discuss the tightness of lower bound \eqref{13}. Fig.~\ref{fig:bound_tightness} presents a direct comparison between the achievable LQR cost and the lower bound. The true performance curve was computed numerically using the semidefinite programming (SDP) method from \cite{Tanaka}, which yields the minimum achievable cost for a given CNER. As illustrated, the gap between achievable performance and the theoretical bound is minimal across different CNER values. This provides strong numerical evidence that the bound is sufficiently tight, justifying our approach of optimizing it as an effective proxy for the true LQR cost.

In addition, the LQR-cost bound \eqref{13} is meaningful only when the system is in the stable region, which requires
\begin{equation}
\label{18}
D^{\mathbf{SC}^3}_k>\log_2|\det \mathbf{A}_k|.
\end{equation}
As shown in \eqref{8}, \eqref{11}, and \eqref{12}, sensor selection influences UL transmission quality and determines computing requirements, while actuator selection affects DL transmission quality and determines the controlled system dynamics. Sensor-actuator pairing is inherently coupled with communication and computing cycle rates. Collectively, these factors jointly determine the CNER and the LQR control cost.
\begin{figure} [t]
\centering
\includegraphics[width=0.9\linewidth]{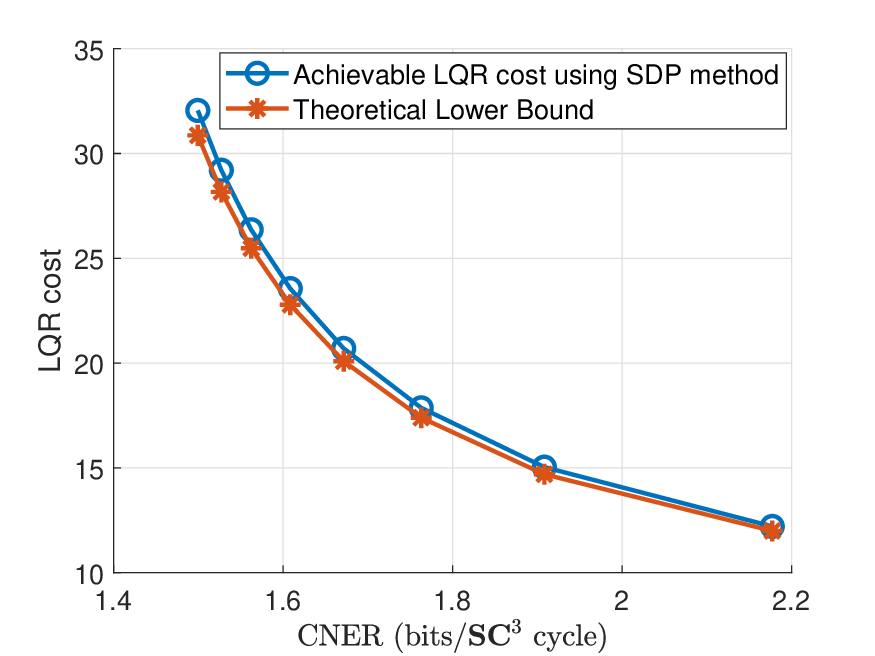}
\caption{Comparison between the achievable LQR cost, computed via the SDP method, and the theoretical lower bound. The parameters are set as follows: $\mathbf{A}=\mathrm{diag}(2, 1.2, 0.3)$, $\mathbf{B}=\mathbf{I}_{3}$, $\mathbf{Q}=\mathbf{I}_{3}$, $\mathbf{R}=\mathbf{I}_{3}$, and covariance $\Sigma_{\mathbf{v}}=\mathbf{I}_{3}$.}
\label{fig:bound_tightness}
\end{figure}

\section{Joint Sensor-Actuator Pairing and Resource Allocation}
\label{section 3}
\subsection{Problem Formulation}
In this section, we introduce the system-level closed-loop optimization problem, which optimizes the pairing decision, UL\&DL bandwidth and transmit power, and the computing CPU frequency to minimize the sum LQR cost of $K$ $\mathbf{SC}^3$ closed loops. The optimization problem is formulated as
\begin{subequations}
\begin{align}
	\mbox{(P1)} \ \ &\min\limits_{\mathcal{A}, \mathcal{P}^u, \mathcal{B}, \mathcal{F}, \mathcal{P}^d} \sum\limits_{k=1}^Kl_k \\
	\text{s.t.} \
	&l_k \geqslant \frac{n_k N \!\left( \mathbf{v}_k\right)|\det \mathbf{M}_k|^\frac{1}{n_k}} {2^{\frac{2}{n_k}(D^{\mathbf{SC}^3}_k-\log_2|\det \mathbf{A}_k|)}-1}+\mathrm{Tr}\left( \mathbf{\Sigma}_{\mathbf{v}_k}\mathbf{S}_k\right), \forall k \label{16b} \\
	&D^{\mathbf{SC}^3}_k> \log_2|\det \mathbf{A}_k|,\forall k \label{a16} \\
	&{D^{\mathbf{SC}^3}_k\leqslant\min\{\sum_{s=1}^Sa_{sk}\rho_s D^u_k, D^d_k\}, \forall k} \label{16c}\\
	&D^u_k\leqslant \sum\limits_{s=1}^Sa_{sk}t^u_kB_k\log_2(1+\frac{|h_s^u|^2p_s^u}{B_kN_0}), \forall k \label{16d} \ \\
	&D^d_k\leqslant t^d_kB_k \log_2(1+\frac{|h_k^d|^2p_k^d}{B_kN_0}), \forall k \label{16e}\\
	&\frac{\sum\limits_{s=1}^Sa_{sk}\gamma_{s}D^u_k}{f_k}\leqslant t^c_k, \forall k \label{16f}\\
	&p_s^u\leqslant p_{s,\max}^u,\forall s \label{16i} \\
	&\sum\limits_{k=1}^K p_k^d\leqslant p_{\max}^d \label{16i-1} \\
	&\sum\limits_{k=1}^K B_k\leqslant B_{\max} \label{16g}\\
	&\sum\limits_{k=1}^Kf_k\leqslant f_{\max} \label{16h}\\
	&B_k\geqslant 0,\ \forall k \\
	&f_k\geqslant 0,\ \forall k \\
	&p_k^d\geqslant 0,\ \forall k \label{16k-1}\\ 
	&p_s^u\geqslant 0,\ \forall s \label{16k}\\
	&\sum\limits_{k=1}^{K}a_{sk}\leqslant1, \forall s \label{16j}\\
	&\sum\limits_{s=1}^Sa_{sk}=1,\forall k \label{16l}\\
	&a_{sk}=0 \ \text{if} \ s\notin \mathcal{S}_k,\ \forall k \label{16t} \\
	&a_{sk} \in \{0, 1\}, \quad \forall s, k \label{16o}
\end{align}
\end{subequations}
{where 
$\mathcal{A} = \{a_{sk}\}_{\forall s,k}$ represents the pairing matrix, while 
$\mathcal{P}^u = \{p^u_s\}_{\forall s}$, $\mathcal{B} = \{B_k\}_{\forall k}$, $\mathcal{F} = \{f_k\}_{\forall k}$, and $\mathcal{P}^d = \{p^d_k\}_{\forall k}$ denote the sets for UL power, bandwidth, computing frequency, and DL power, respectively. $B_{\max}$, $p_{s,\max}^u$, $p_{\max}^d$, and $f_{\max}$ denote the maximum available bandwidth, maximum UL transmit power of sensor $s$, maximum DL transmit power, and maximum CPU frequency, respectively. 
Within the formulated optimization framework, constraint \eqref{16b} characterizes the theoretical lower bound of the LQR control cost, while \eqref{a16} specifies the stability requirement for the closed-loop system. Furthermore, \eqref{16c} ensures that the CNER is bounded by the respective UL and DL cycle rates, effectively capturing the coupling between UL and DL. Constraints \eqref{16d} and \eqref{16e} represent the transmission rate limits for UL and DL, respectively, while \eqref{16f} accounts for the computing latency budget at the EIH. Additionally, \eqref{16i}--\eqref{16k} impose resource constraints, and \eqref{16j}--\eqref{16o} enforce one-to-one sensor-actuator pairing and sensing range limitations.}

Evidently, (P1) is a complex MINLP problem with tightly coupled integer and continuous variables due to task-level interdependencies among $\mathbf{SC}^3$ functions. Employing exhaustive search to find the optimal solution would incur prohibitive computational complexity. For mission-critical applications, decision-making speed is essential to task efficiency. To address this challenge, we introduce a LOAC framework inspired by \cite{Huang}, which effectively integrates learning with optimization techniques to solve the MINLP problem efficiently.

\subsection{Learning-Optimization-Integrated Actor-Critic Framework}

\begin{figure*} [t]
\centering
\includegraphics[width=0.85\linewidth]{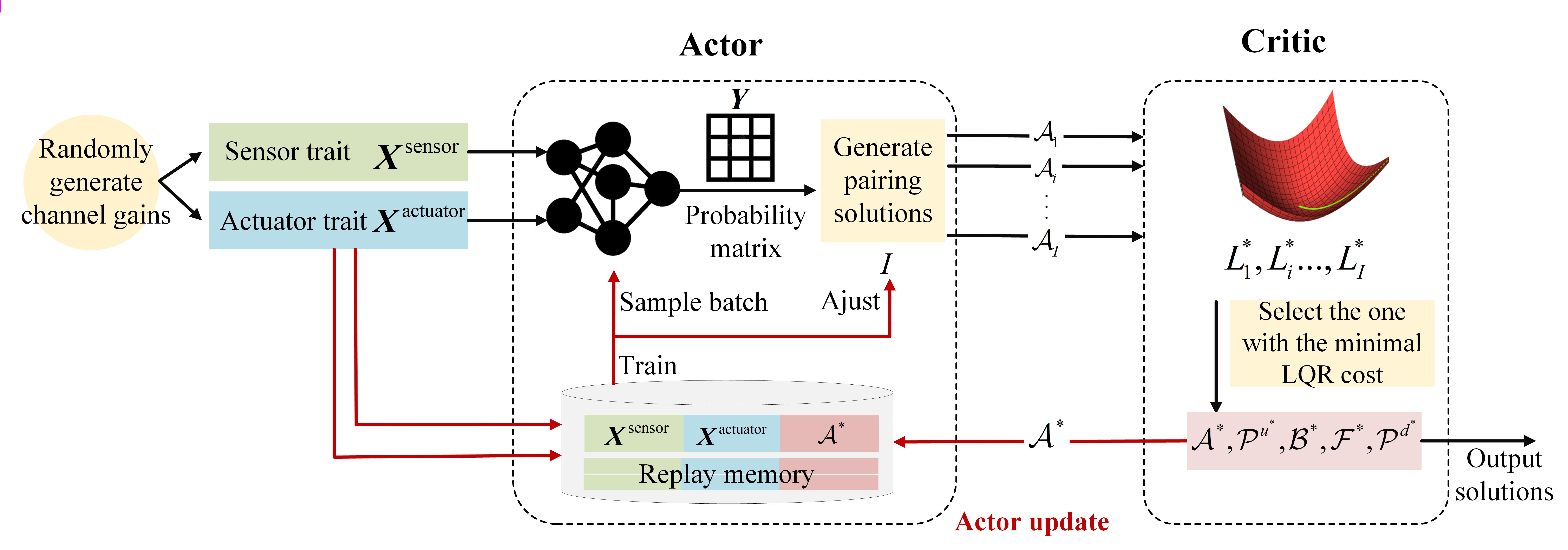}
\caption{The structure of the proposed LOAC framework, which includes a DNN-based actor and an optimization-based critic.} 
\label{LOAC}
\end{figure*}

The structure of the proposed LOAC framework is illustrated in Fig.~\ref{LOAC}. In this framework, a DNN-based actor generates sensor-actuator pairing decisions. Given a specific pairing decision, the optimization-based critic solves the resulting continuous resource allocation problem. To capture the interdependence between pairing and resource allocation, the critic evaluates pairing decision quality based on the achieved LQR cost. This evaluation is subsequently fed back to the actor to iteratively refine the DNN, thereby improving pairing decision quality over successive iterations. In the following, we introduce the LOAC framework in detail.

\subsubsection{Actor Module} The actor module comprises a DNN and a pairing generation algorithm. The DNN consists of three parts: a sensor-trait encoder, an actuator-trait encoder, and a pairing network.
The sensor-trait encoder takes as input the traits of $S$ sensors, denoted as $\mathbf{X}^{\text{sensor}}$, where each row corresponds to the trait vector of an individual sensor, as defined in \eqref{1}, and outputs the encoded sensor traits. Similarly, the actuator-trait encoder takes as input the traits of $K$ actuators, denoted as $\mathbf{X}^{\text{actuator}}$, where each row represents the trait vector of an individual actuator, as defined in \eqref{2}, and outputs the encoded actuator traits. The pairing network takes the encoded sensor and actuator traits as input and produces the raw pairing probability matrix, denoted as $\mathbf{Y}^{\text{raw}} = [y_{sk}^{\text{raw}}]$ with $y_{sk}^{\text{raw}} \ (0\leqslant y_{sk}^{\text{raw}}\leqslant 1)$ representing the pairing probability between sensor $s$ and actuator $k$.

To convert these scores into $I$ valid, one-to-one pairing solutions, denoted as $\mathcal{A}_i = \{a_{sk}^i\}$, we employ a sequential selection procedure detailed in {\textbf{Algorithm~\ref{alg:matching}}}. This procedure ensures that each sensor is assigned to at most one actuator.
To mitigate ordering bias in the selection process, the order of the $K$ actuators is randomly shuffled at the beginning of each solution generation attempt. Following this shuffled order $\mathcal{K}_{\text{shuffled}}$, each actuator $k$ sequentially selects a sensor based on a corrected probability distribution derived from $y_{sk}^{\text{raw}}$. Formally,
the probability of actuator $k$ selecting sensor $s$ is given by
\begin{equation}
\label{eq:pairing_prob}
y_{sk} = \frac{ y_{sk}^{\text{raw}} \cdot \mathbf{1}_{\{s \in \mathcal{S}_k \setminus \mathcal{S}_{\text{selected}}\}} }{ \sum_{j \in \mathcal{S}_k} \left( y_{jk}^{\text{raw}} \cdot \mathbf{1}_{\{j \notin \mathcal{S}_{\text{selected}}\}} \right) },
\end{equation}
where $\mathbf{1}_{\{\cdot\}}$ is the indicator function, which is 1 if the condition is true and 0 otherwise, and
$\mathcal{S}_{\text{selected}}$ is the set of selected sensors. Assuming the selected sensor index of actuator $k$ is $s_k$, then the selected sensor set is updated by
\begin{equation}
\label{21}
\mathcal{S}_{\text{selected}} = \mathcal{S}_{\text{selected}} \cup \{s_k\}.
\end{equation}
If at any step an actuator has no available sensors (i.e., the denominator in \eqref{eq:pairing_prob} vanishes), the current generation attempt is terminated and a new attempt begins. This process repeats until I  complete pairing solutions are successfully generated.

\begin{algorithm}[t]
\caption{Algorithm for Generating Pairing Solution} \small
\label{alg:matching}
\begin{algorithmic}[1]
\REQUIRE
{
	Number of pairing solutions to generate: $I$; \\
	Pairing score matrix: $\mathbf{Y}^{\text{raw}} = [y_{sk}^{\text{raw}}] \in \mathbb{R}^{S \times K}$; \\
	The compatible sensor set for each actuator: $\mathcal{S}_k, \ \forall k \in \{1, \dots, K\}$;\\
	Maximum number of attempts per solution: $N_{\text{max}}$.
}
\ENSURE
Set of $I$ valid pairing solutions: $\mathcal{A}$.

\STATE \textit{Initialize}: $\mathcal{A}\leftarrow \emptyset$;
\FOR{$i = 1$ to $I$}
\FOR{$n = 1$ to $N_{\text{max}}$}
\STATE \textit{Initialize for new attempt}: $\mathcal{S}_{\text{selected}} \leftarrow \emptyset$;
\STATE Temporary pairing solution $\mathcal{A}_i=\{a_{sk}^i\} \leftarrow \{0\}^{S\times K};$
\STATE Randomly shuffle actuator indices to get $\mathcal{K}_{\text{shuffled}}$;
\STATE valid\_attempt $\leftarrow$ \textbf{true};
\FOR{$k$ in $\mathcal{K}_{\text{shuffled}}$} 
\STATE Calculate the probability $\mathbf{y}_k = [y_{sk}]$ using \eqref{eq:pairing_prob};
\IF{$\sum_{s \in \mathcal{S}_k} y_{sk} = 0$} 
\STATE valid\_attempt $\leftarrow$ \textbf{false};
\STATE \textbf{break};
\ENDIF
\STATE Sample a sensor $s_k$ according to the distribution $\mathbf{y}_k$;
\STATE Update the selected set according to \eqref{21};
\STATE $a_{s_kk}^i \leftarrow 1$;
\ENDFOR
\IF{valid\_attempt \textbf{and} $\mathcal{A}_i \notin \mathcal{A}$} 
\STATE $\mathcal{A} \leftarrow \mathcal{A} \cup \{\mathcal{A}_i\}$;
\STATE success $\leftarrow$ \textbf{true};
\STATE \textbf{break}; 
\ENDIF
\ENDFOR
\ENDFOR
\end{algorithmic}
\end{algorithm}

\subsubsection{Critic Module}
For a given sensor-actuator pairing, (P1) is reduced to a continuous optimization problem. For simplicity, let $s_k$ denote the index of the sensor paired with actuator $k$. The problem can be reformulated as
\begin{subequations}
\begin{align}
\mbox{(P2)} \ \ &\min\limits_{\mathcal{P}^u, \mathcal{B},\mathcal{F},\mathcal{P}^d} \sum\limits_{k=1}^Kl_k \\
\text{s.t.} \ &l_k \geqslant\frac{n_k N \!\left( \mathbf{v}_k\right)|\det \mathbf{M}_k|^\frac{1}{n_k}} {2^{\frac{2}{n_k}(D^{\mathbf{SC}^3}_k-\log_2|\det \mathbf{A}_k|)}-1}+\mathrm{Tr}\left( \mathbf{\Sigma}_{\mathbf{v}_k}\mathbf{S}_k\right),\ \forall k \label{20b} \\
&D^{\mathbf{SC}^3}_k> \log_2|\det \mathbf{A}_k|, \ \forall k\\
&D^{\mathbf{SC}^3}_k\leqslant\min\{\rho_{s_k}D^u_k, D^d_k\}, \ \forall k \label{20c}\\
&D^u_k\leqslant t^u_k B_k\log_2(1+\frac{|h_{s_k}^u|^2p_{s_k}^u}{B_kN_0}),\ \forall k \label{20e} \ \\
&D^d_k\leqslant t^d_k B_k\log_2(1+\frac{|h_k^d|^2p_k^d}{B_kN_0}), \ \forall k \label{20ef}\\
&\frac{\gamma_{s_k}D^u_k}{f_k}\leqslant t^c_k, \ \forall k \label{20f}\\
& \eqref{16i}-\eqref{16k}.
\end{align}
\end{subequations}
(P2) is a non-convex problem with tightly coupled optimization variables of UL\&DL transmit power, bandwidth, and computing CPU frequency. To solve this problem, we first determine the optimal value of the UL transmit power
\begin{equation}
(p^u_{s_k})^*=p_{s_k,\max}^{u}, \forall k.
\end{equation}
This follows from the fact that the LQR cost lower bound in \eqref{20b} is a monotonically decreasing function of CNER. Furthermore, as defined in \eqref{20c}, CNER is determined by the performance bottleneck—specifically, the minimum among the cycle rates of UL transmission, edge computing, and DL transmission. Given that UL cycle rate increases monotonically with sensor transmit power, the optimal strategy is for the selected sensor to transmit at maximum power, thereby maximizing CNER and consequently minimizing the resulting LQR cost.

In addition, the computing-latency constraint \eqref{20f} can be equivalently transformed into a linear constraint as
\begin{equation}
\gamma_{s_k}D^u_k\leqslant t^c_kf_k, \forall k.
\end{equation}
On this basis, we obtain
\begin{subequations}
\begin{align}
\mbox{(P3)} \ \ &\min\limits_{\mathcal{B},\mathcal{F},\mathcal{P}^d} \sum\limits_{k=1}^Kl_k \\
\text{s.t.} \ &l_k \geqslant \frac{n_k N \!\left( \mathbf{v}_k\right)|\det \mathbf{M}_k|^\frac{1}{n_k}} {2^{\frac{2}{n_k}(D^{\mathbf{SC}^3}_k-\log_2|\det \mathbf{A}_k|)}-1}+\mathrm{Tr}\left( \mathbf{\Sigma}_{\mathbf{v}_k}\mathbf{S}_k\right), \forall k \label{25b} \\
&D^{\mathbf{SC}^3}_k> \log_2|\det \mathbf{A}_k|, \forall k \label{a25}\\
&{D^{\mathbf{SC}^3}_k\leqslant\min\{\rho_{s_k}D^u_k, D^d_k\}, \ \forall k} \label{25c}\\
&D^u_k\leqslant t^u_k B_k\log_2(1+\frac{|h_{s_k}^u|^2p_{s_k, \max}^{u}}{B_kN_0}),\ \forall k \label{25e} \ \\
&D^d_k\leqslant t^d_k B_k\log_2(1+\frac{|h_k^d|^2p_k^d}{B_kN_0}),\ \forall k \label{25f}\\
&{\gamma_{s_k}D^u_k\leqslant t^c_kf_k, \ \forall k} \label{25g}\\
&\eqref{16i-1}--\eqref{16k-1} \label{25k},
\end{align}
\end{subequations}
where the CNER constraint \eqref{25c} is relaxed into an inequality. This relaxation does not alter the optimal solution, as the optimum is achieved when CNER attains its maximum allowable value, rendering the inequality tight at optimality.
\begin{theorem}
\label{the}
The resource allocation problem (P3) is convex.
\end{theorem}
\begin{proof}
See Appendix.
\end{proof}
Based on {\bf{Theorem \ref{the}}}, (P3) can be efficiently solved using convex optimization tools.

\begin{algorithm}[htbp]
\caption{Algorithm in Learning-Optimization-Integrated Actor-Critic Framework}
\label{alg:LOAC}
\begin{algorithmic}[1]
\REQUIRE
Number of pairing solutions $I$, the maximum training epochs $E$, batch size $G$, and learning rate $lr$;\\
The EIH location and sensor traits \eqref{1} as well as actuator traits \eqref{2} excluding channel gains;\\
Communication parameters: $B_{\max}$, $p_{\max}^d$, $f_c$, $c$, $\eta_{\text{LOS}}$, $\eta_{\text{NLOS}}$, $a$, $b$, $t^u_k$, and $t^d_k$;\\
Computing parameters: $f_{\max}$ and $t^c_k$;\\
Control parameters: $\mathbf{A}_k$, $\mathbf{B}_k$, $\mathbf{Q}_k$, $\mathbf{R}_k$, $n_k$, $m_k$, and the distribution of the system noise $\mathbf{v}_k$.
\STATE Calculate the effective sensor set $\mathcal{S}_k$ according to \eqref{4}.
\STATE Calculate the large-scale fading for each sensor and actuator according to \eqref{6};
\STATE \textbf{Initialize:} replay buffer $\mathcal{D}\leftarrow \emptyset$;
\FOR{epoch $e = 1$ to $E$}
\STATE Randomly generate UL\&DL small-scale fading (i.i.d. $\mathcal{CN}(0,1)$) and compute channel gains according to \eqref{5}.
\STATE Construct and normalize the trait matrices $\mathbf{X}^{\text{sensor}}$ and $\mathbf{X}^{\text{actuator}}$;
\STATE Input $\mathbf{X}^{\text{sensor}}$ and $\mathbf{X}^{\text{actuator}}$ into the DNN to obtain the pairing probability matrix $\mathbf{Y}$;
\STATE Generate pairing candidates $\mathcal{A}=\{\mathcal{A}_i\}_{i=1}^I$ using \textbf{Algorithm~\ref{alg:matching}};
\FOR{each candidate $\mathcal{A}_{i}$}
\STATE Solve the resource allocation problem (P3) to obtain the LQR cost and the associated allocation solution.
\ENDFOR
\STATE Select the best solution according to \eqref{26};
\STATE Store experience tuple $(\mathbf{X}^{\text{sensor}}, \mathbf{X}^{\text{actuator}}, \mathcal{A}^*)$ in $\mathcal{D}$;
\IF {$|\mathcal{D}|\geqslant G$}
\STATE Sample a mini-batch from $\mathcal{D}$ and update the DNN;
\ENDIF
\STATE Decrease $I$ and $lr$ gradually as $e$ increases;
\ENDFOR
\ENSURE Pairing solution $\mathcal{A}^*$ and resource allocation results: $\mathcal{P}^{u*}$, $\mathcal{B}^*$, $\mathcal{F}^*$, $\mathcal{P}^{d*}$.
\end{algorithmic}
\end{algorithm}

\subsubsection{Actor-Critic Interaction}
{To effectively supervise the DNN-generated candidate solutions, the critic evaluates each pairing solution, i.e., $\mathcal{A}_i$, by solving the continuous resource allocation problem (P3). The corresponding LQR cost for the $i$-th pairing solution, denoted as $L_i^*$, is computed as:
\begin{equation}
\label{25}
L_i^* = \min\limits_{\mathcal{P}^u, \mathcal{B}, \mathcal{F}, \mathcal{P}^d} \sum_{k=1}^K l_{k} \ \Bigg|_{\mathcal{A} = \mathcal{A}_i}.
\end{equation}
Subsequently, the critic selects the pairing solution that achieves the lowest overall LQR cost. This optimal decision is then fed back to the actor, which is mathematically expressed as:
\begin{equation}
\label{26}
\begin{aligned}
	i^* &= \arg\min_i L_i^*, \\
	\big(\mathcal{A}^*, \mathcal{P}^{u*}, \mathcal{B}^*, \mathcal{F}^*, \mathcal{P}^{d*}\big) &= \big(\mathcal{A}_{i^*}, \mathcal{P}^u_{i^*}, \mathcal{B}_{i^*}, \mathcal{F}_{i^*}, \mathcal{P}^d_{i^*}\big).
\end{aligned}
\end{equation}}
Then, in the actor module, the matrices of sensor and actuator traits as well as the pairing solution provided by the critic form an experience tuple, $\{\mathbf{X}^{\text{sensor}},\mathbf{X}^{\text{actuator}},\mathcal{A}^{*}\}$, which is stored in the replay buffer, denoted as $\mathcal{D}$.
For each channel realization, the DNN is updated using a mini-batch sampled from the replay buffer. To enhance training efficiency, the number of candidate pairing solutions $I$ is gradually reduced throughout the training process. Notably, since the critic module solves a convex optimization problem and yields the globally optimal LQR cost for each pairing solution, its evaluations are highly accurate. This precise feedback not only ensures fast convergence but also effectively guides the DNN toward an optimal policy. The complete algorithm in the proposed LOAC framework is summarized in {\bf Algorithm~\ref{alg:LOAC}}.

{
\subsubsection{Computational Complexity Analysis}
\label{complexity}
To rigorously evaluate the computational efficiency of the proposed LOAC scheme, we provide a comparative complexity analysis against two benchmarks: exhaustive search and the optimization-based SCA scheme. For a standard comparison, the complexity of solving the continuous convex sub-problems is evaluated assuming the adoption of the standard interior-point method (IPM). According to \cite{CovOpt}, the complexity of an IPM for a convex problem with $N_{\text{var}}$ optimization variables and $N_{\text{cons}}$ inequality constraints is $\mathcal{O}\left(\left(N_{\text{var}} + N_{\text{cons}}\right)^{3.5} \log\left(1/\epsilon\right)\right)$, where $\epsilon$ denotes the duality gap threshold. On this basis, the complexity profiles of the three methods are delineated as follows:
\begin{itemize}
\item Proposed LOAC Scheme: The complexity of the LOAC scheme is formulated as $\mathcal{O}\left(LM^2 + \left(12K + 3\right)^{3.5} \log\left(1/\epsilon\right)\right)$, where $L$ is the number of DNN layers and $M$ is the maximum number of neurons per layer. The first term, $LM^2$, represents the computational overhead of the DNN's feed-forward inference used to predict the optimal sensor-actuator pairing. The second term accounts for solving the resulting pairing-fixed convex sub-problem, which involves $3K$ variables and $9K + 3$ constraints. By decoupling the combinatorial pairing from resource allocation, the LOAC scheme achieves polynomial-time complexity.

\item Exhaustive Search: This scheme serves as the performance upper bound by exploring all possible sensor-actuator permutations. Its complexity is expressed as $\mathcal{O}\left(_SP_K\left(12K + 3\right)^{3.5} \log\left(1/\epsilon\right)\right)$, where $_S P_K$ denotes the permutations of $K$ elements from $S$ available sensors. Due to the factorial growth of $_S P_K$, this approach becomes computationally prohibitive as the network scale increases, making it impractical for real-time 6G orchestration.

\item Optimization-based Scheme: This scheme relaxes the binary pairing variables $a_{sk} \in \{0, 1\}$ into continuous variables $\widetilde{a}_{sk} \in [0, 1]$ and applies SCA to solve the resulting continuous non-convex optimization problem \cite{survey3}. 
The relaxed problem includes $SK + 3K + S$ variables and $SK + 3S + 11K + 3$ constraints, leading to a complexity of $\mathcal{O}\left(J \left(2SK + 4S + 14K + 3\right)^{3.5} \log\left(1/\epsilon\right)\right)$, where $J$ denotes the number of iterations. For the recovery of binary solutions from the continuous results, heuristic recovery mechanisms such as simple threshold rounding or greedy selection based on probability rankings are typically employed \cite{Zhang2019MaxMin}. Given that the number of sensors significantly exceeds the number of actuators, i.e., $S \gg K$, in practical 6G deployments, the complexity of this algorithm is heavily driven by the product $SK$. 
\end{itemize}
}

\section{Simulation Results and Discussion}
\label{section 4}
\begin{figure} [t]
\centering
\includegraphics[width=0.9\linewidth]{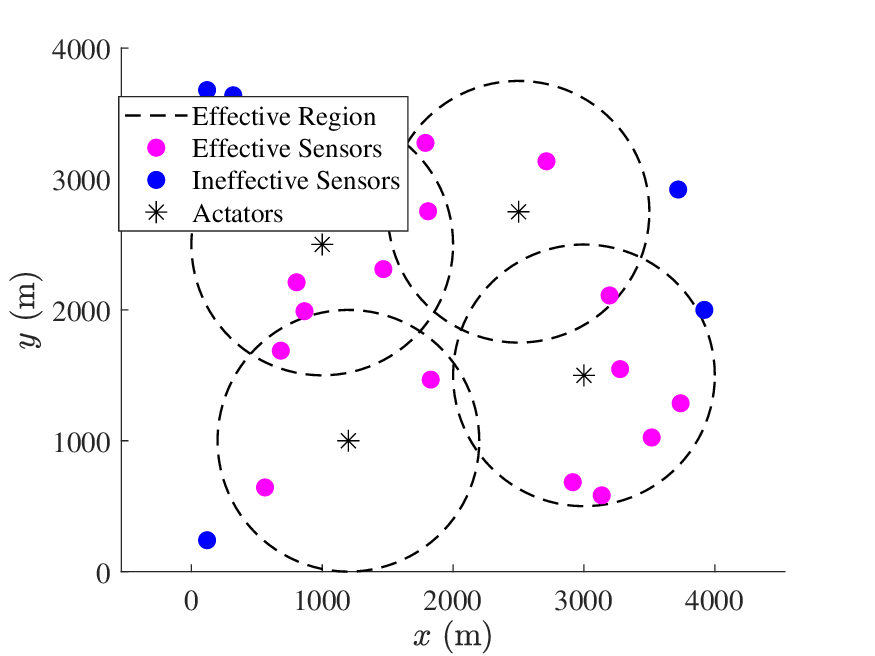}
\caption{Topology of the 6G-enabled autonomous operation system, including $S=20$ sensors and $K=4$ actuators.}
\label{topo}
\end{figure}
\begin{figure*} [t]
\centering
\includegraphics[width=0.65\linewidth]{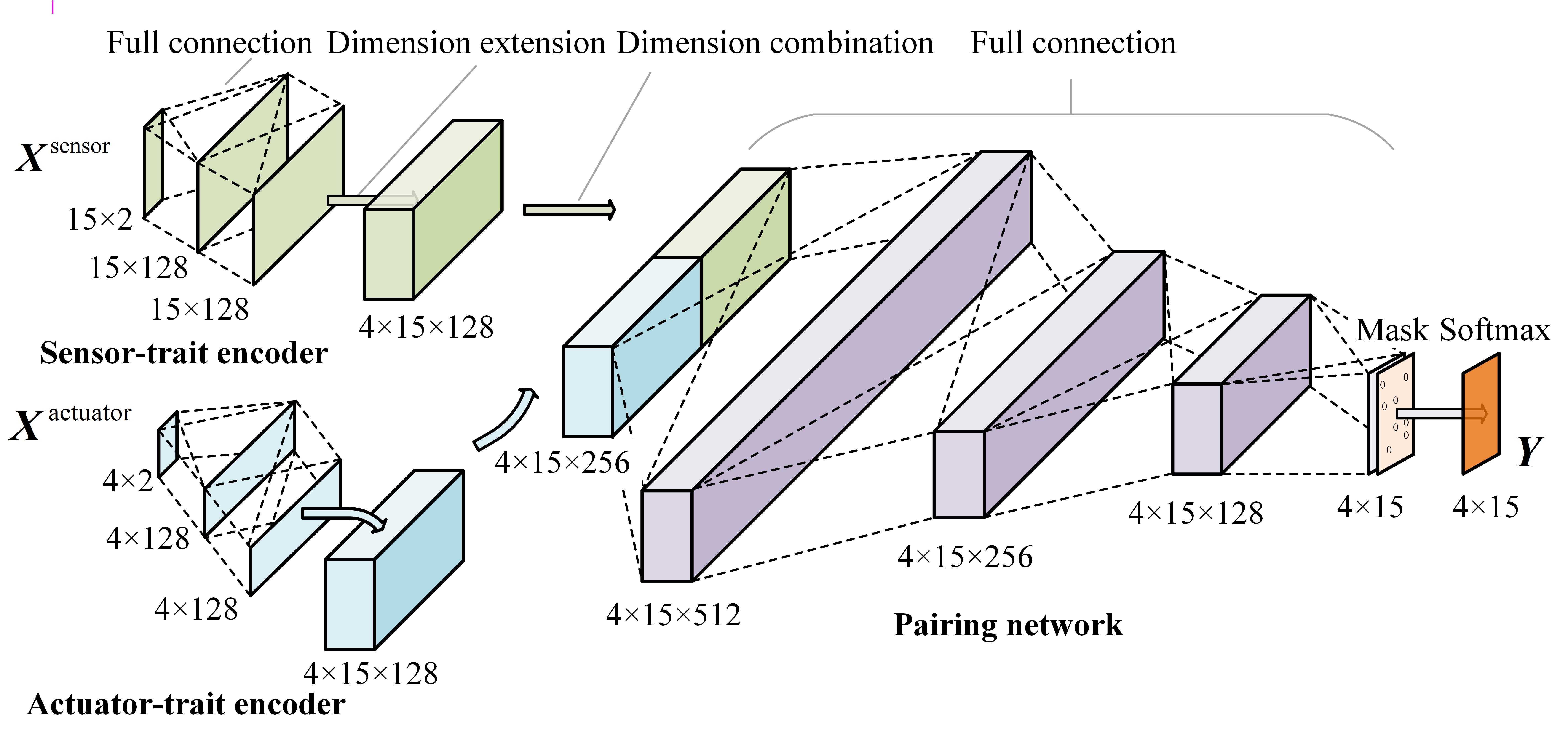}
\caption{Structure of the DNN, including the sensor-trait encoder, the actuator-trait encoder, and the pairing network.}
\label{DNN}
\end{figure*}

We present the simulation topology in Fig.~\ref{topo}. The layout is a $4000$ m $\times$ $4000$ m square, with the EIH positioned at $[2000, 2000]$ m and a height of $h=300$ m. The system comprises $K=4$ actuators and $S=20$ sensors, organized into $K=4$ $\mathbf{SC}^3$ closed loops, each responsible for an independent control task. Each sensor has a circular sensing range of $r_s=1000$ m \cite{sensing}. By calculating Euclidean distances between sensors and actuators, effective sensors, those with at least one actuator within their sensing range, can be identified. In our topology, $S=15$ effective sensors are highlighted in pink; the remaining $5$ sensors are excluded from pairing as no actuators operate within their ranges.
The topology is fixed in our simulations. In practical autonomous systems, sensor and actuator positions may vary over time. The proposed LOAC framework naturally adapts to such variations: the critic continuously evaluates pairing quality and feeds back updated assessments to the actor, enabling online adjustment of the pairing network. Upon detecting topology changes, the number of candidate pairing solutions $I$ and the learning rate $lr$ can temporarily revert to larger values, allowing the LOAC network to re-explore the solution space and stabilize under the new configuration. The rollback procedure requires flexible design for different application scenarios, which is beyond the scope of this paper.

We summarize the simulation parameters in Table~\ref{tab1}. In this simulation, we consider a special case where $\mathbf{R}_k=\mathbf{0}$, which emphasizes the state deviation cost while ignoring the control cost \cite{Kostina}. Under this assumption, the computation of $\mathbf{M}_k$ in~\eqref{17} simplifies to
\begin{equation}
\mathbf{M}_k
= \mathbf{S}_k \mathbf{B}_k \big( \mathbf{B}_k^{\mathrm{T}} \mathbf{S}_k \mathbf{B}_k \big)^{-1} \mathbf{B}_k^{\mathrm{T}} \mathbf{S}_k.
\end{equation}
If $\mathbf{B}_k^{\mathrm{T}} \mathbf{S}_k \mathbf{B}_k$ is invertible, then the optimal solution satisfies $\mathbf{M}_k = \mathbf{S}_k$. Substituting this into the Riccati recursion yields
\begin{equation}
\mathbf{S}_k = \mathbf{Q}_k + \mathbf{A}_k^{\mathrm{T}} \big( \mathbf{S}_k - \mathbf{M}_k \big) \mathbf{A}_k
= \mathbf{Q}_k.
\end{equation}
Therefore, for this special case, we obtain
\begin{equation}
\mathbf{S}_k = \mathbf{M}_k = \mathbf{Q}_k,
\end{equation}
which becomes independent of the specific values of $\mathbf{A}_k$ and $\mathbf{B}_k$, provided that $\mathbf{B}_k$, $\mathbf{S}_k$, and $\mathbf{B}_k^{\mathrm{T}} \mathbf{S}_k \mathbf{B}_k$ are invertible. For this reason, the explicit values of $\mathbf{A}_k$ and $\mathbf{B}_k$ are not specified in this simulation case.

\begin{table}[t]
\centering
\caption{Simulation Parameters}
\label{tab1}
\begin{tabular}{p{2cm}p{6.2cm}}
\toprule
\textbf{Category} & \textbf{Parameter Description} \\
\midrule
\multicolumn{2}{c}{\textbf{Communication Parameters}} \\
\midrule
$p^{u}_{s,\max}$ & Maximum UL transmit power of each sensor: $0.1$ W \\
$p^d_{\max}$ & Maximum DL transmit power of the EIH: $4$ W \\
$f_c$ & Carrier frequency: $2000$ MHz \\
$c$ & Speed of light: $3 \times 10^8$ m/s \\
$\eta_{\text{LOS}}$, $\eta_{\text{NLOS}}$, $a$, $b$
& Environmental parameters: $0.1$, $21$, $5.0188$, $0.3511$ \\
$t^u_k$ & UL transmission time: $4$ ms, $\forall k$ \\
$t^d_k$ & DL transmission time: $1$ ms, $\forall k$ \\
\midrule
\multicolumn{2}{c}{\textbf{Computing Parameters}} \\
\midrule
$f_{\max}$ & The maximal computing CPU frequency: $1$ GHz\\
$\gamma_s$ & CPU cycles per bit (randomly generated from a uniform distribution of $[50, 500]$) \\
$\rho_s$ & Information extraction ratio: $0.01$, $\forall s$ \\
$t^c_k$ & Computing time: $4$ ms, $\forall k$ \\
\midrule
\multicolumn{2}{c}{\textbf{Control Parameters}}\\
\midrule
$n_k$, $m_k$ & State and input dimensions: $100$, $\forall k$ \\
$\log_2|\det\mathbf{A}_k|$ & Intrinsic entropy: $[10, 100, 60, 40]$ \\
$\mathbf{R}_k$ & Input cost matrix: $\mathbf{0}_{100}$ \\
$\mathbf{Q}_k$ & State cost matrix: $\mathbf{I}_{100}$ \\
$\mathbf{v}_k$ & System noise: white Gaussian distribution with $\mathbf{\Sigma}_{\mathbf{v}_k}=0.01\times\mathbf{I}_{100}$, $\forall k$ \\
\bottomrule
\end{tabular}
\end{table}
Fig.~\ref{DNN} illustrates the architecture of the DNN employed in our simulation. Given that all sensors have the same sensing range, maximum transmit power, and information extraction ratio, the distinguishing characteristics are used as the input to the sensor-trait encoder, i.e., the UL channel gain and the CPU frequency required to process one-bit sensing data. Therefore, the $s$-th row of the input matrix $\mathbf{X}^{\text{sensor}} \in \mathbb{R}^{S \times 2}$ is $\left[|h^u_s|^2, \gamma_s\right]$.
Similarly, as we set the same value for $\mathbf{Q}_k$, $\mathbf{R}_k$, and $\mathbf{\Sigma}_{\mathbf{v}_k}$ across all controlled systems, the actuator-trait encoder takes the DL channel gain and the intrinsic entropy of the associated controlled system as the input, i.e., $\mathbf{X}^{\text{actuator}} \in \mathbb{R}^{K \times 2}$, where the $k$-th row is $\left[|h^d_k|^2, \log_2|\det\mathbf{A}_k|\right]$. Prior to feeding $\mathbf{X}^{\text{sensor}}$ and $\mathbf{X}^{\text{actuator}}$ into the DNN, the intrinsic entropy is normalized through linear min-max scaling, while the other parameters are normalized via logarithmic min-max scaling. The sensor-trait and actuator-trait encoders include a single fully connected hidden layer with a LeakyReLU activation function (negative slope coefficient of 0.1). Their output dimensions are $S \times 128$ and $K \times 128$, respectively. Subsequently, the encoded sensor and actuator traits are expanded to dimensions $S \times K \times 128$ and concatenated along the trait axis, resulting in a joint tensor of shape $S \times K \times 256$. This tensor is then passed through the pairing network, which consists of three fully connected hidden layers with output dimensions $S \times K \times [512, 256, 128]$, each followed by a LeakyReLU activation. The output of the final layer is a matrix $\mathbf{Y}^{\text{raw}} = [y_{sk}^{\text{raw}}] \in \mathbb{R}^{S \times K}$, where each entry $y_{sk}^{\text{raw}}$ represents the unnormalized pairing score between sensor $s$ and actuator $k$. A column-wise softmax function is applied to $\mathbf{Y}^{\text{raw}}$ to obtain a probability matrix. To ensure the sensing-range constraint, for any actuator $k$ and sensor $s \notin \mathcal{S}_k$ (defined in \eqref{4}), the corresponding probability $y_{sk}$ is set to zero. As a result, the pairing probability matrix is corrected to satisfy,
\begin{equation}
\sum\limits_{s=1}^S y_{sk} = 1, \quad \text{with} \quad y_{sk} = 0 \ \text{if} \ s \notin \mathcal{S}_k.
\end{equation}

\begin{figure} [t]
\centering
\includegraphics[width=0.9\linewidth]{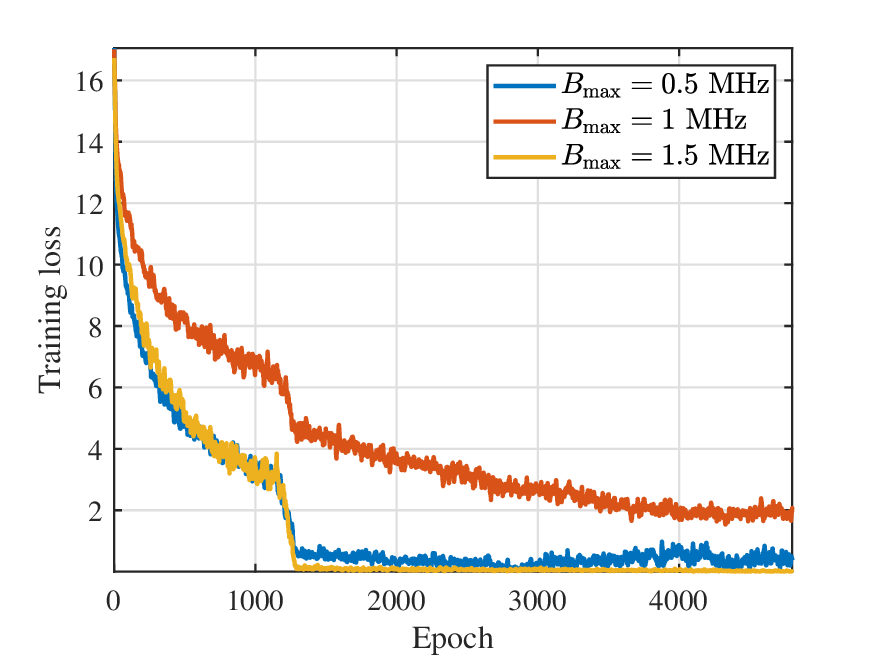}
\caption{Training loss of the DNN under different bandwidth conditions.}
\label{Fig7}
\end{figure}

During the training stage, the small-scale fading is randomly generated according to a complex Gaussian distribution, $\mathcal{CN}(0,1)$. For each channel realization, the DNN calculates a pairing probability matrix $\mathbf{Y}$ and generates $I$ candidate solutions using {\bf{Algorithm} 1}. The initial number of candidate pairing solutions is set to $I = 32$ and is halved every $512$ training epochs, with a minimum of $I = 2$. The maximum number of attempts for generating each pairing solution is $100$. The training process employs a replay buffer with a capacity of $1280$ and a batch size of $128$. We adopt the Adam optimizer with an initial learning rate of $10^{-3}$, which is decayed by a factor of $\frac{1}{\sqrt{2}}$ every 256 epochs to facilitate stable convergence. The cross-entropy loss function is applied to train the DNN.

Fig.~\ref{Fig7} presents the training loss curve over $4800$ epochs, where the loss values are smoothed using a moving average with a window size of 10 for clarity. The results demonstrate that the network converges fast. More importantly, under extreme bandwidth conditions, when $B_{\max} = 0.5$ MHz (limited) or $B_{\max} = 1.5$ MHz (adequate), the proposed scheme converges within 1500 epochs. This rapid convergence underscores the learning efficiency of the LOAC framework, wherein the critic module provides effective feedback that guides the DNN toward optimal decision-making over time.

\begin{figure} [t]
\centering
\includegraphics[width=0.8\linewidth]{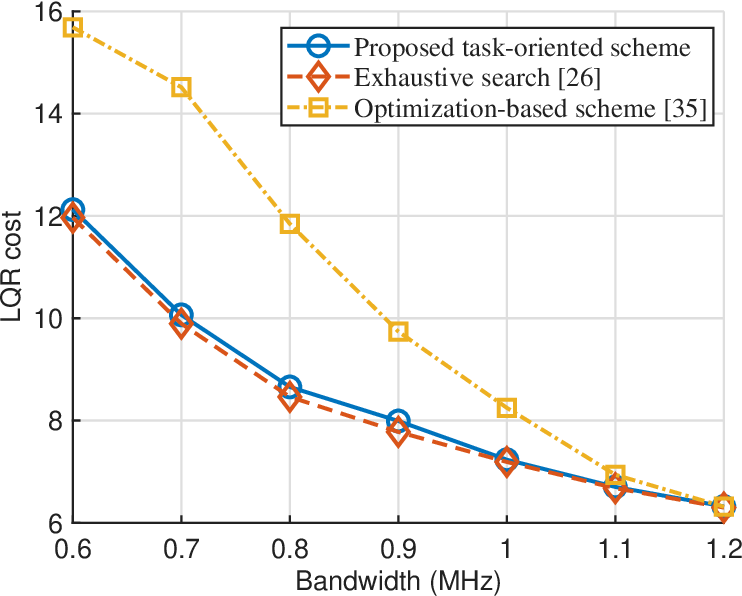}
\caption{{LQR cost achieved by the proposed scheme, exhaustive search, and the optimization-based baseline.}}
\label{Fig9}
\end{figure}
In the following simulations, all the results are averaged over 100 independently generated random channel realizations to ensure statistical reliability. 

Fig.~\ref{Fig9} compares the performance of different algorithms for solving the MINLP problem. Specifically, the proposed LOAC scheme is evaluated against exhaustive search \cite{Li2026Joint} and an optimization-based scheme \cite{Zhang2019MaxMin}. The optimal solution, serving as the performance upper bound, is obtained by exhaustively enumerating all feasible sensor-actuator pairings to minimize the LQR cost. Meanwhile, the optimization-based scheme relaxes integer pairing variables into continuous values, applies SCA to tackle the resulting non-convex problem, and recovers a feasible discrete pairing via index-ordered greedy selection. Simulation results demonstrate that the proposed LOAC scheme achieves near-optimal control performance, closely approaching that of exhaustive search. In contrast, the optimization-based scheme exhibits considerable performance loss, primarily stemming from continuous relaxation and discrete rounding operations. Furthermore, as detailed in Section~\ref{complexity}, exhaustive search incurs prohibitive factorial complexity, whereas the proposed LOAC scheme maintains low polynomial-time complexity. Consequently, the LOAC scheme emerges as an effective, engineering-friendly algorithm for solving MINLP problems in time-varying environments.

\begin{figure} [t]
\centering
\includegraphics[width=0.8\linewidth]{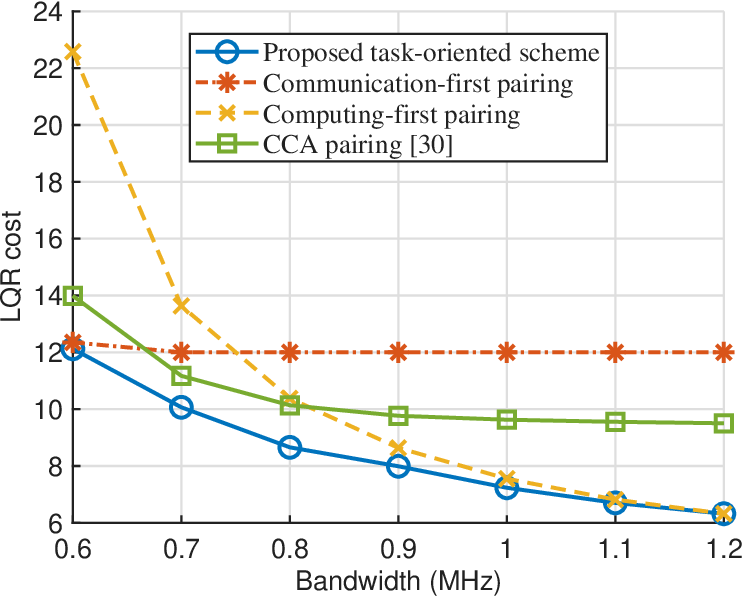}
\caption{{LQR cost comparison between the proposed scheme and three representative pairing baselines.}}
\label{Fig8}
\end{figure}

{In Fig. \ref{Fig8}, we evaluate the proposed scheme against three heuristic pairing schemes: communication-first pairing, computing-first pairing, and communication-computation-aware (CCA) pairing \cite{Gu2023Communication}. For a fair comparison, the continuous resource allocation for all baselines is obtained by solving (P3). These baseline schemes employ the Hungarian algorithm to find the optimal matching based on a designed utility matrix, i.e., $\mathbf{U} \in \mathbb{R}^{S \times K}$. Each element in the utility matrix, denoted by $u_{sk}$, represents the pairing score between sensor $s$ and actuator $k$. First, actuators are prioritized based on the instability of their assigned control tasks. The score of actuator $k$, denoted by $u_{k}^{\text{actuator}}$, is defined by the normalized intrinsic entropy:
\begin{equation}
u_{k}^{\text{actuator}} = \frac{\log_2|\det\mathbf{A}_k|}{\sum\limits_{j=1}^K \log_2|\det\mathbf{A}_j|}, \quad \forall k.
\end{equation}
Next, depending on the pairing schemes, the score of each candidate sensor $s$, denoted by $u_{sk}^{\text{sensor}}$, is mathematically formulated as follows:
\begin{itemize}
\item Communication-first pairing: Focuses solely on the channel conditions, given by
\begin{equation}
u_{sk}^{\text{sensor}} = \frac{|h^u_s|^2}{\sum\limits_{j \in \mathcal{S}_k} |h^u_j|^2}.
\end{equation}
\item Computing-first pairing: Focuses solely on the computational load, given by
\begin{equation}
u_{sk}^{\text{sensor}} = \frac{1/\gamma_s}{\sum\limits_{j \in \mathcal{S}_k} 1/\gamma_j}.
\end{equation}
\item CCA pairing: Balances communication and computing metrics via a weighting hyperparameter $\omega \in [0,1]$, given by
\begin{equation}
u_{sk}^{\text{sensor}} = \omega \frac{|h^u_s|^2}{\sum\limits_{j \in \mathcal{S}_k} |h^u_j|^2} + (1-\omega) \frac{1/\gamma_s}{\sum\limits_{j \in \mathcal{S}_k} 1/\gamma_j}.
\end{equation}
\end{itemize}
The elements of the overall utility matrix $\mathbf{U}$ are then constructed as:
\begin{equation}
u_{sk} = 
\begin{cases} 
u_{sk}^{\text{sensor}} \times u_{k}^{\text{actuator}}, & \text{if } s \in \mathcal{S}_k, \\ 
0, & \text{otherwise}. 
\end{cases}
\end{equation}
It is observed that the proposed scheme consistently outperforms the three benchmark pairing schemes. In the bandwidth-limited regime, communication-first pairing performs near-optimally, as overcoming the transmission bottleneck is the primary concern. However, as available bandwidth increases, the communication bottleneck relaxes, and pairing priority smoothly shifts toward computing-first pairing to alleviate computational load at the EIH. For CCA pairing, the hyperparameter $\omega \in [0,1]$ governs this trade-off: $\omega=1$ yields pure communication-first pairing, while $\omega=0$ reduces to pure computing-first pairing. The specific setting of $\omega=0.3$ provides intermediate performance between these two extremes. These results reveal a fundamental insight: the essence of sensor-actuator pairing lies in dynamically balancing communication and computing capabilities to prevent any single bottleneck from constraining overall closed-loop performance. By treating the $\mathbf{SC}^3$ closed loop as an integrated unit, the proposed DNN-based pairing automatically achieves this dynamic balance.}

{
\begin{figure} [t]
\centering
\includegraphics[width=0.8\linewidth]{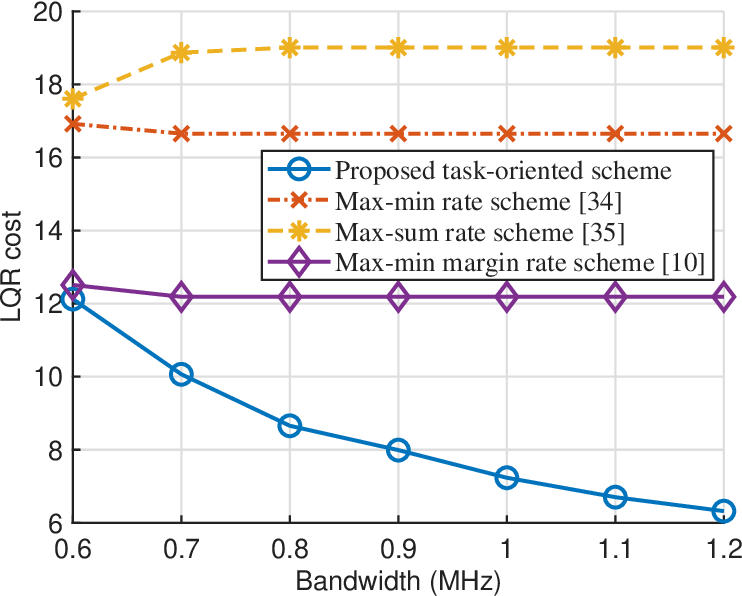}
\caption{{LQR cost comparison between the proposed scheme and communication-oriented benchmarks.}}
\label{fig5}
\end{figure}
In Fig. \ref{fig5}, we compare the proposed LOAC scheme with three communication-oriented schemes: the max-sum rate scheme \cite{Sekander2017Decoupled}, the max-min rate scheme \cite{Zhang2019MaxMin}, and the max-min margin rate scheme \cite{Fang1}. To maximize the communication efficiency, the three compared schemes adopt the communication-first pairing strategy. For a fair comparison, we adapt these schemes to jointly optimize bandwidth, transmit power, and CPU frequency within our framework. These schemes maximize a rate-based objective and include a control constraint to ensure basic stability: $D^{\mathbf{SC}^3}_k \geqslant \log_2|\det\mathbf{A}_k| + \Delta_k$. Here, $\Delta_k=15$ (bits/$\mathbf{SC}^3$ cycle) is a predefined stability margin. 
The objectives of the three benchmarks are defined as follows:
\begin{itemize}
\item Max-Sum Rate Scheme: Maximizes the total CNER across all $\mathbf{SC}^3$ closed loops, i.e., $\max \sum\limits_{k=1}^K D^{\mathbf{SC}^3}_k$.
\item Max-Min Rate Scheme: Maximizes the worst-case CNER to ensure fairness, i.e., $\max \min\limits_{k} D^{\mathbf{SC}^3}_k$.
\item Max-Min Margin Rate Scheme: Maximizes the worst-case margin rate by subtracting the intrinsic entropy, i.e., $\max \min\limits_{k} \left( D^{\mathbf{SC}^3}_k - \log_2|\det\mathbf{A}_k| \right)$.
\end{itemize}

Simulation results demonstrate that the proposed scheme markedly outperforms all three baselines. Notably, the LQR cost of the max-sum-rate scheme increases with additional bandwidth. This anomaly stems from the greedy nature of the sum-rate objective: to maximize total rate, the optimizer allocates most shared resources to $\mathbf{SC}^3$ closed loops with favorable communication and computing capabilities, forcing the remaining loops to operate near their stability margins. Since LQR cost increases sharply near the stability boundary but decreases only marginally within the stable region, this imbalance ultimately inflates the total cost.
The max-min-rate scheme considers all $\mathbf{SC}^3$
closed loops to ensure rate fairness. However, different control tasks require different control efforts; equalizing absolute rates leaves highly unstable tasks barely stable while naturally stable tasks receive surplus resources. Consequently, the LQR cost remains high and shows minimal improvement with increased bandwidth.
Regarding the max-min-margin-rate scheme, it accounts for varying control requirements by aligning its objective with task instability. Consequently, it achieves significantly better performance compared to the two traditional rate-based schemes. The remaining performance gap in the high-bandwidth region primarily stems from its communication-first pairing strategy. This scheme serves as a strong heuristic alternative for minimizing task-oriented LQR cost, consistent with the theoretical analysis in \cite{Fang1}.
In summary, these results underscore the importance of task-oriented design for improving resource efficiency in autonomous operation systems. Our scheme integrates task performance directly into closed-loop design, effectively bridging the gap between task-agnostic rate maximization and actual physical control performance.
}

{ 
\begin{figure}[t]
\centering
\includegraphics[width=0.8\linewidth]{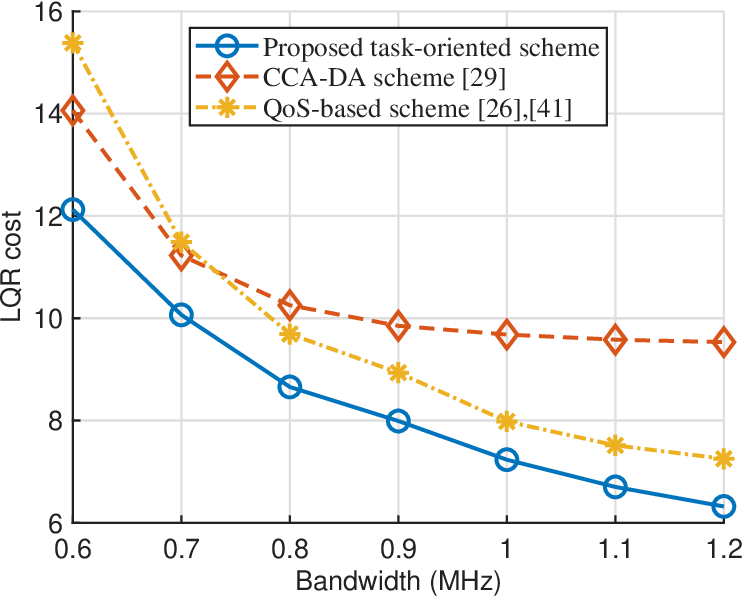}
\caption{{LQR cost comparison between the proposed scheme and control-aware baselines.}}
\label{fig7}
\end{figure}

In Fig.~\ref{fig7}, we demonstrate the superiority of the proposed scheme by comparing it against two task-aware baselines: a quality-of-service (QoS)-based scheme \cite{Li2026Joint,varma2024resource} and a CCA pairing with decoupled allocation (CCA-DA) scheme \cite{Tran}. The implementation details are as follows:
\begin{enumerate}
	\item \textbf{QoS-based scheme} \cite{Li2026Joint,varma2024resource}: Inspired by 5G QoS allocation paradigms \cite{varma2024resource}, this scheme utilizes the intrinsic entropy of control tasks to determine proportional allocation ratios for continuous resources, given by $\frac{\log_2|\det(\mathbf{A}_k)|}{\sum\limits_{j=1}^K\log_2|\det\mathbf{A}_j|}$. The optimal discrete pairing solution is subsequently obtained through exhaustive search \cite{Li2026Joint}.
	\item \textbf{CCA-DA scheme} \cite{Tran}: Representing the widely adopted decoupled methodology for tackling MINLPs, this approach separates the optimization of pairing and resource allocation. The CCA pairing scheme with weight $\omega=0.3$ is applied to pair sensors and actuators. Subsequently, continuous communication and computing resources are optimized independently, assuming that the counterpart domain possesses infinite cycle rates.
\end{enumerate}
As demonstrated by simulation results, the proposed scheme significantly outperforms the QoS-based baseline. The latter leverages the concept of 5G differentiated services, allocating resources based on coarse-grained task requirements. However, this approach leads to a mismatch between dynamic task demands and assigned resources, preventing their effective conversion into task-level efficiency. While subsequent exhaustive search prevents further performance degradation, it incurs prohibitive computational complexity. Regarding the CCA-DA scheme, its performance gap relative to the proposed scheme exposes a fundamental limitation of decoupled designs: they neglect the interplay between communication and computing, as well as the intrinsic coupling between sensor-actuator pairing and resource allocation. Consequently, $\mathbf{SC}^3$ closed-loop performance is constrained by a bottleneck in either communication or computing, leaving counterpart resources underutilized. Ultimately, these comparisons underscore the necessity of treating the $\mathbf{SC}^3$ closed loop as an integrated entity. Through joint pairing and resource allocation, our scheme efficiently utilizes all available communication and computing resources to maximize control performance.}

\begin{figure} [t]
\centering
\includegraphics[width=0.9\linewidth]{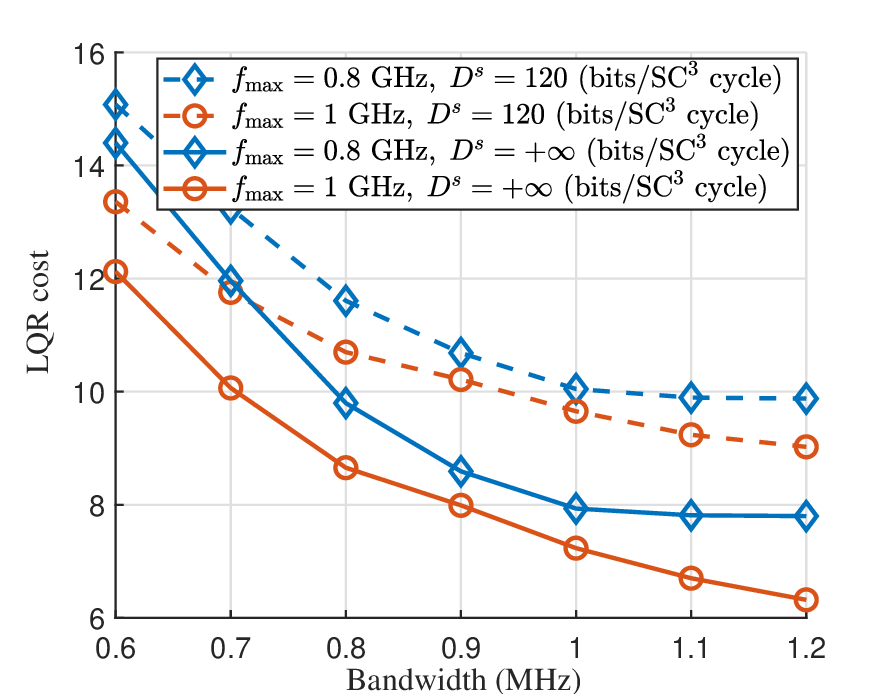}
\caption{LQR cost under different sensing cycle rates and CPU frequency constraints.}
\label{Fig6}
\end{figure}

Fig.~\ref{Fig6} presents the LQR cost under varying sensing cycle rates and CPU frequency constraints. For simplicity, we assume that all sensors share the same sensing cycle rate, denoted as $D^{s}$. The results show that when the sensing cycle rate is finite, the LQR cost increases as it becomes the bottleneck. This occurs because the limited sensing rate reduces the achievable CNER, thereby impairing control performance. Furthermore, reducing the maximum available CPU frequency also elevates the LQR cost. When $f_{\max} = 0.8$~GHz, the LQR cost saturates once bandwidth exceeds approximately $1$~MHz, indicating that computing becomes the performance bottleneck in this regime. These observations highlight the intrinsic interdependence within $\mathbf{SC}^3$ closed loops: balanced sensing, communication, and computing capabilities are essential for maintaining desirable closed-loop performance. This, in turn, underscores the necessity of the proposed joint pairing and resource allocation strategy.

\section{Conclusions}
\label{section 5}
In this paper, we have investigated a 6G-enabled autonomous operation system. We proposed a system-level closed-loop optimization scheme that jointly optimizes task orchestration (sensor-actuator pairing) and resource scheduling (communication and computing resource allocation) to minimize the sum LQR cost. To address the resulting MINLP problem, we developed an LOAC framework that achieves near-optimal solutions with low computational complexity. Simulation results validated the superiority of the proposed scheme, demonstrating that it adaptively aligns sensor and actuator characteristics with available communication and computing capabilities to form efficient $\mathbf{SC}^3$ closed loops. By enabling dynamic task-resource bidirectional adaptation, the proposed scheme opens new avenues for 6G network design, fundamentally differing from traditional methods that target either network performance or isolated task optimization.

\appendix
\begin{center}
	{Proof of \textbf{Theorem 1}}
\end{center}

According to the definition in \cite{CovOpt}, a problem can be identified convex if it has a convex objective and convex constraints. In (P3), the objective and constraints \eqref{a25}, \eqref{25g}--\eqref{25k} are linear. Therefore, we only need to prove the convexity of the LQR-cost constraint \eqref{25b}, the CNER constraint \eqref{25c}, and the UL\&DL cycle rate constraints \eqref{25e} and \eqref{25f}. For the LQR-cost constraint, we use a function, $g(D^{\mathbf{SC}^3}_k)$, to represent its right-hand side expression
\begin{equation}
	g(D^{\mathbf{SC}^3}_k)\triangleq \frac{n_k N \!\left( \mathbf{v}_k\right)|\det(\mathbf{M}_k)|^\frac{1}{n_k}} {2^{\frac{2}{n_k}(D^{\mathbf{SC}^3}_k-\log_2|\det(\mathbf{A}_k)|)}-1}+\mathrm{Tr}\left( \mathbf{\Sigma}_{\mathbf{v}_k}\mathbf{S}_k\right).
\end{equation}
On this basis, we calculate the second-order derivative of $g(D^{\mathbf{SC}^3}_k)$
\begin{equation}
	\begin{aligned}
		&\frac{\partial^2 g(D^{\mathbf{SC}^3}_k)}{\partial (D^{\mathbf{SC}^3}_k)^2}=\frac{4N(\mathbf{v}_k)|\det(\mathbf{M}_k)|^{\frac{1}{n_k}}}{n_k(\log_2e)^2}\times\\
		&\quad\quad\frac{2^{\frac{2}{n_k}(D^{\mathbf{SC}^3}_k-\log_2|\det(\mathbf{A}_k)|)}\big[2^{\frac{2}{n_k}(D^{\mathbf{SC}^3}_k-\log_2|\det(\mathbf{A}_k)|)}+1\big]}{\big[2^{\frac{2}{n_k}(D^{\mathbf{SC}^3}_k-\log_2|\det(\mathbf{A}_k)|)}-1\big]^3},
	\end{aligned}
\end{equation}
which is greater than zero as long as the stable condition, $D^{\mathbf{SC}^3}_k > \log_2|\det(\mathbf{A}_k)|$, is satisfied. Therefore, the LQR-cost constraint is convex in the stable region. As for the CNER constraint \eqref{25c}, the constraint involving the minimum function is convex because it can be equivalently rewritten as a set of linear constraints
\begin{equation}
	D^{\mathbf{SC}^3}_k \leqslant \min\{\rho_{s_k}D^u_k, \ D^d_k\} \Rightarrow 
	\begin{cases}
		D^{\mathbf{SC}^3}_k \leqslant \rho_{s_k}D^u_k, \\
		D^{\mathbf{SC}^3}_k \leqslant D^d_k,
	\end{cases}
\end{equation}
which are all linear constraints. Then, for the UL\& DL cycle rate constraints \eqref{25e} and \eqref{25f}, the bandwidth and transmit power are expressed in the form of $g(x,y)=ax\log_2\left(1+b\frac{y}{x}\right)$, where $a$ and $b$ are constants. We calculate the Hessian matrix of this function, denoted as $\mathbf{H}_g(x, y)$,
\begin{equation}
	\mathbf{H}_g(x, y)
	=
	\begin{bmatrix}
		-\dfrac{a b^2 y^2}{x (b y + x)^2 \ln 2} & \dfrac{a b^2 y}{(b y + x)^2 \ln 2} \\[10pt]
		\dfrac{a b^2 y}{(b y + x)^2 \ln 2} & -\dfrac{a b^2 x}{(b y + x)^2 \ln 2}
	\end{bmatrix}.
\end{equation}
For a symmetric $2\times2$ matrix, it is negative semidefinite if and only if its diagonal entries are non-positive and its determinant is non-negative. Under the condition that $x,y,b>0$, we have
\begin{equation}
	\begin{aligned}
		& -\frac{a b^2 y^2}{x (b y + x)^2 \ln 2}<0, \\
		&-\dfrac{a b^2 x}{(b y + x)^2 \ln 2}<0,
	\end{aligned}
\end{equation}
and
\begin{equation}
	\begin{vmatrix}
		-\dfrac{a b^{2} y^{2}}{x (b y + x)^{2} \ln 2} \quad
		&
		\dfrac{a b^{2} y}{(b y + x)^{2} \ln 2}
		\\[15pt]
		\dfrac{a b^{2} y}{(b y + x)^{2} \ln 2} \quad
		&
		-\dfrac{a b^{2} x}{(b y + x)^{2} \ln 2}
	\end{vmatrix}=0.
\end{equation}
Therefore, the Hessian matrix $\mathbf{H}_g(x, y)$ is negative semidefinite, indicating that the function $g(x, y) = a x \log_2\left(1 + b \frac{y}{x} \right)$ is concave. Consequently, the UL\&DL cycle rate constraints in \eqref{25e} and \eqref{25f} adhere to the principles of convex optimization. As both the objective function and all constraints are convex, (P3) is a convex optimization problem.

\balance

\end{document}